\documentclass[12pt]{article}
\usepackage[top=1in,left=1in,right=1in,bottom=1in]{geometry}

\usepackage{style}  
\usepackage{fullpage}
\usepackage{times}

\usepackage[compact]{titlesec}

\newcommand{\wcost}{\omega}

\newcommand{\id}[1]{\ifmmode\mathit{#1}\else\textit{#1}\fi}
\newcommand{\const}[1]{\ifmmode\mbox{\textc{#1}}\else\textsc{#1}\fi}
\newcommand{\ourmodel}{Asymmetric NP}
\newcommand{\ourmodellong}{Asymmetric Nested-Parallel}

\newcommand{\seqmodel}{Asymmetric RAM}
\newcommand{\depth}{depth}
\newcommand{\local}{symmetric}
\newcommand{\implicit}{implicit $k$-decomposition}
\newcommand{\first}{\mb{first}}
\newcommand{\last}{\mb{last}}
\newcommand{\low}{\mb{low}}
\newcommand{\high}{\mb{high}}
\newcommand{\imprep}{BC labeling}

\newcommand{\clustergraph}{clusters graph}
\newcommand{\localgraph}{local graph}
\newcommand{\outver}{outside vertex}
\newcommand{\outvers}{outside vertices}

\renewcommand{\whp}[1]{\textit{whp}}

\usepackage{microtype}
\title{Implicit Decomposition for \\Write-Efficient Connectivity Algorithms}

\author{Naama Ben-David$^\ast$ \qquad Guy E. Blelloch$^\ast$ \qquad
Jeremy T. Fineman$^\dag$\smallskip\\
Phillip B. Gibbons$^\ast$ \qquad Yan Gu$^\ast$ \qquad
Charles McGuffey$^\ast$ \qquad Julian Shun$^\ddag$\bigskip\\
$^\ast$Carnegie Mellon University \quad $^\dag$Georgetown University
\quad $^\ddag$UC Berkeley\bigskip\\
\{nbendavi, guyb, gibbons, yan.gu, cmcguffe\}@cs.cmu.edu\smallskip\\
jfineman@cs.georgetown.edu \quad jshun@eecs.berkeley.edu
}

\date{}

\begin{document}
    \maketitle
    \thispagestyle{empty}
    \par
\setcounter{secnumdepth}{2}
\setcounter{tocdepth}{2}

\begin{abstract}
The future of main memory appears to lie in the direction of new
technologies that provide strong capacity-to-performance ratios, but
have write operations that are much more expensive than reads in terms
of latency, bandwidth, and energy. Motivated by this trend, we propose
sequential and parallel algorithms to solve graph connectivity problems using
significantly fewer writes than conventional algorithms. Our primary
algorithmic tool is the construction of an $o(n)$-sized {\em implicit
decomposition} of a bounded-degree graph $G$ on $n$ nodes, which combined with read-only access to
$G$ enables fast answers to connectivity and biconnectivity queries on $G$.
The construction breaks the linear-write ``barrier'',
resulting in costs that are asymptotically lower than conventional algorithms while adding only a modest cost to querying time.
For general non-sparse graphs on $m$ edges,
we also provide the first $o(m)$ writes and $O(m)$ operations
parallel algorithms for connectivity and biconnectivity.
These algorithms provide insight into how
applications can efficiently process computations on large graphs in systems
with read-write asymmetry.
\end{abstract}

\hide{
\vfill
\noindent
Corresponding Author: Yan Gu, yan.gu@cs.cmu.edu
}

\newpage
\setcounter{page}{1}

\hide{  Old intro.
The future of main memory appears to lie in a new wave of nonvolatile
memory technologies (e.g., phase-change memory, spin-torque transfer
magnetic RAM, memristor-based resistive RAM, conductive-bridging RAM)
that promise persistence, significantly lower energy costs, and higher
density than the DRAM technology used in today's main
memories~\cite{hp-nvm15, intel-nvm15, Meena14, Yole13}.  A key
property of such technologies, however, is their asymmetric read-write
costs: writes can be an order of magnitude or more higher energy,
higher latency, lower (per-module) bandwidth, and more wear-out prone than
reads~\cite{Akel11, Athanassoulis12, blelloch2016efficient, BFGGS15,
  carson2016write, ChoL09, Dong09, Dong08, HuZXTGS14, ibm-pcm14b,
  Kim14, LeeIMB09, Qureshi12, Xu11, yang:iscas07, ZhouZYZ09, ZWT13}.
Moreover, because bits are stored in these technologies as ``at rest''
states of the given material that can be quickly read but require
physical change to update, this asymmetry appears
fundamental.\footnote{See Appendix~\ref{sec:hardware} for more
technical details of the new memories.}  This motivates the need for
algorithms that are \emph{write-efficient}, in that they significantly
reduce the number of writes compared to existing algorithms.

Models of computation for memories with asymmetric read-write costs
have been proposed and studied by a number of recent
papers~\cite{BT06, Eppstein14, Gal05, ParkS09, Chen11, Viglas12,
  Viglas14, BFGGS15, blelloch2016efficient, BBFGGMS16,
  carson2016write, jacob2017} (see Appendix~\ref{sec:priorwork} for details).  In
this paper, we focus on two such models that are simple enough for
algorithm design while still capturing the read-write asymmetry: (i)
the \emph{\seqmodel{} model}~\cite{blelloch2016efficient}, in which
writes to the asymmetric memory cost $\wcost \gg 1$ and all other
operations are unit cost; and (ii) its parallel variant, the
\emph{\ourmodel} model~\cite{BBFGGMS16}.  Both models have a small
\local{} memory that can be used to help minimize the number of writes
to the large asymmetric memory.

\hide{
The series of work by Blelloch et al.~\cite{BFGGS15,blelloch2016efficient,BBFGGMS16}
defines the cost models, the \emph{Asymmetric RAM model}, of the new hardware that is intuitively similar to the RAM model, except the existence of a small symmetric memory and each write to the main memory costs $\wcost{}$, a non-constant parameter representing the cost of a write relative to a read.
This model is adopted by many recent works~\cite{BFGGS15,blelloch2016efficient,BBFGGMS16,jacob2017}, and the lower and upper bounds on lots of classic problems are shown, as well as the parallel algorithms on these problems (analyzed in a similar model named as \emph{Asymmetric NP model}).
Carson et al.~\cite{carson2016write} presented write-efficient sequential
algorithms for a similar model, as well as write-efficient parallel
algorithms (and lower bounds) on a distributed memory model with
asymmetric read-write costs, focusing on linear algebra problems and
direct N-body methods.
}


We present the first write-efficient sequential and parallel
algorithms for graph connectivity (connected components, spanning
forests) and biconnectivity (biconnected components, articulation
points, and related 1-edge-connectivity) problems. The algorithms
significantly reduce the number of writes to the asymmetric memory
compared to existing algorithms, improving the overall sequential time
and parallel work bounds.  In the algorithms we often cannot afford to
write out the result since just writing it requires too many writes.  Instead we
construct a compact \emph{``oracle''} that can answer queries (e.g. if two
vertices are connected) efficiently.  The costs for constructing the
oracles and query times are summarized in Table~\ref{tbl:results}.
}

\section{Introduction}
\label{sec:intro}

Recent trends in computer memory technologies suggest wide adoption
of memory systems in which reading from memory is significantly
cheaper than writing to it, especially with regards to energy.  The
reason for this asymmetry, roughly speaking, is that writing to memory
requires a change to the state of the material, while reading only
requires detecting the current state.\footnote{See Appendix~\ref{sec:hardware}
for some technical details of the new memories.}  This trend poses the
interesting question of how to design algorithms that are more
efficient in the number of writes than in the number of reads.  To
this end recent works have studied models that account for asymmetric
read-write costs and have analyzed a variety of algorithms in these
models~\cite{BT06, BBFGGMS16, BFGGS15, blelloch2016efficient,
carson2016write, Chen11, Eppstein14, Gal05, jacob2017, ParkS09,
Viglas12, Viglas14}.

Some of this work has shown an inherent tradeoff between reads and
writes.  For example, Blelloch et al.~\cite{blelloch2016efficient}
show that for computations on certain DAGs, any reduction in writes
requires an increase in reads.  For FFTs and sorting networks, for
example, the increase in reads is exponential with respect to the
decrease in writes.  Intuitively, the tradeoff is because reducing
writes restricts the ability to save partial results needed in
multiple places, and hence requires repeating certain computations.
This is reminiscent of well-studied time-space
tradeoffs~\cite{HK81}---but not the same, because time-space models
still allow an \emph{arbitrary} number of writes to the limited space,
with each write costing the same as a read.

In this paper we are interested in graph connectivity problems, and in
particular whether it is possible to build an ``oracle'' using a
sublinear number of writes that supports fast queries, along with any
read-write tradeoffs this entails.  We consider undirected connectivity
(connected components, spanning forests) and biconnectivity
(biconnected components, articulation points, and related
1-edge-connectivity) problems.  We do not consider the cost of
initially storing the graph in memory, but note that there are many
scenarios in which the graph is either represented implicitly, e.g.,
the Swendsen-Wang algorithm~\cite{swendsen1987nonuniversal}, or for
which the graph is sampled and used multiple times, e.g., edges
selected based on different Boolean hash functions or based on
properties (timestamp, weight, relationship, etc.) associated with the
edge.

Our results show that if a graph with $n$ vertices and $m$ edges is
sufficiently dense, a sublinear number of writes ($o(m)$) can be
achieved without asymptotically increasing the number of reads (no
tradeoff is required).  For bounded-degree graphs, on the other hand,
our algorithm achieving a sublinear number of writes ($o(n)$) involves
a linear tradeoff between reads and writes.  The main technical
contribution is a new \emph{implicit} decomposition of a graph that
allows writing out information only for a suitably small \emph{sample}
of the vertices.  We use two models to account for the read-write
asymmetry: (i) the \emph{\seqmodel{}
model}~\cite{blelloch2016efficient}, in which writes to the asymmetric
memory cost $\wcost \gg 1$ and all other operations are unit cost; and
(ii) its parallel variant, the \emph{\ourmodel}
model~\cite{BBFGGMS16}.  Both models have a small \local{} memory (a
cache) that can be used to help minimize the number of writes to the
large asymmetric memory.

Table~\ref{tbl:results} summarizes our main results for these models,
showing asymptotic improvements in construction costs over all prior
work (sequential or parallel) for these well-studied connectivity
problems.

\myparagraph{Algorithms with $o(m)$ writes for non-sparse graphs}
The first contribution of this paper is a group of algorithms that achieve
$O(m/\wcost{} + n)$ writes, $O(m)$ other operations, and hence $O(m +
\wcost n)$ work.  While standard sequential BFS- or DFS-based
graph connectivity algorithms require only $O(n)$ writes, and
hence already achieve this bound, the parallel setting is more challenging.
Existing linear-work parallel connectivity algorithms perform
$\Theta(m)$ writes~\cite{ColeKT96, Gazit1991, HalperinZ96, Halperin00,
PettieR02, PoonR97, Shun2014}, and hence are actually
$\Theta(\wcost m)$ work in the asymmetric memory setting.  We show how
the algorithm of Shun et al.~\cite{Shun2014} can be adapted to use
only $O(m/\wcost + n)$ writes (and $O(m)$ other operations), by
avoiding repeated graph contractions and using a recent
write-efficient BFS~\cite{BBFGGMS16}, yielding the first
$O(m+\wcost{}n)$ expected work, low-\depth{} parallel algorithm for
connectivity in the asymmetric setting.  (By \defn{low \depth} we mean
\depth{} polynomial in $\wcost \log n$.)

For biconnectivity, the standard output is an array of size $m$
indicating to which biconnected component each edge
belongs~\cite{CLRS, JaJa92}.  Producing this output requires at least
$m$ writes, and as a result, the sequential time (and parallel work)
ends up being $\Theta(\wcost m)$ in the asymmetric memory setting.  We
present an equally effective representation of the output, which we
call the \emph{\imprep}, which has size only $O(n)$.  This leads to a
sequential biconnectivity algorithm that constructs the oracle in only
$O(m+\wcost n)$ time in the asymmetric setting.  Moreover, we show how
to leverage our new parallel connectivity algorithm to compute the
\imprep{} in $O(m/\wcost + n)$ writes, yielding the first
$O(m+\wcost{}n)$ work parallel algorithm for biconnectivity in the
asymmetric memory setting.  We show:

\begin{theorem}\label{thm:main-linear}
Graph connectivity and biconnectivity oracles can be constructed in
parallel with $O(m+\wcost{}n)$ expected work and $O(\wcost^2\log^2 n)$
\depth{} \whp{}\footnote{Throughout the paper we use \whp{} to mean
with probability $1 - {n^{-c}}$ for any constant $c$ that shows up
linearly in the cost bound (e.g. $O(c \wcost^2 \log^2 n)$ in the
bound given).} on the \ourmodel{} model, and each query can be
answered in $O(1)$ work.  Sequentially, the construction takes $O(m
+ \wcost n)$ time on the \seqmodel{} model, with $O(1)$ query time.
The \local{} memory is $O(\wcost \log n)$ words.
\end{theorem}

\begin{table*}[t]
\begin{center}
\begin{tabular}{c|c@{  }@{  }c|c@{  }@{  }c|c}\toprule
& \multicolumn{2}{c|}{Connectivity} & \multicolumn{2}{c|}{Biconnectivity}
  & Best choice \\
& Seq.~time & Parallel~work & Seq.~time & Parallel~work
  & when \\ \hline
Prior work
  & $O(m+\wcost n)$ & $O(\wcost m)^{\dag}$
  & $O(\wcost m)$ & $O(\wcost m)^{\dag}$
  & -- \\
Ours [\S\ref{sec:cc-linear},\S\ref{sec:biconn-linear}]
  & $O(m + \wcost n)^{\dag}$ & $O(m + \wcost n)^{\dag}$
  & $O(m + \wcost n)^{\dag}$ & $O(m + \wcost n)^{\dag}$
  & $m \in \Omega(\sqrt{\wcost} n)$ \\
Ours [\S\ref{sec:cc-sublinear},\S\ref{sec:biconn-sublinear}]
  & $O(\sqrt{\wcost} m)^{\dag}$ & $O(\sqrt{\wcost} m)^{\dag}$
  & $O(\sqrt{\wcost} m)^{\dag}$ & $O(\sqrt{\wcost} m)^{\dag}$
  & $m \in o(\sqrt{\wcost} n)$ \\ \bottomrule
\end{tabular}
\caption{\label{tbl:results} Summary of main results for constructing
  connectivity oracles ($n$ nodes, $m$ edges, $\dag$=expected), where
  $\wcost \gg 1$ is the cost of writes to the asymmetric memory.
  Compared to prior work, asymmetric memory writes are reduced by up
  to a factor of $\wcost$, yielding improvements in both sequential
  time and parallel work.  All parallel algorithms have \depth{}
  polynomial in $\wcost \log n$.  Query times are
  $O(\sqrt{\wcost})^{\dag}$ (connectivity) and $O(\wcost)^{\dag}$
  (biconnectivity) for the last row and $O(1)$ for the rest.
  For all
  algorithms the small \local{} memory is only $O(\wcost \log n)$ words.}
\vspace{-0.1in}
\end{center}
\end{table*}

\myparagraph{Algorithms with $o(n)$ writes for sparse graphs}
For sparse graphs, the work of our connectivity and
biconnectivity algorithms is dominated by the $\Omega(n)$ writes they
perform.  This led us to explore the
following fundamental question: \defn{Is it possible to construct,
using $o(n)$ writes to the asymmetric memory, an oracle for graph
connectivity (or biconnectivity) that can answer queries in time
independent of $n$?}  Given that the standard output for these
problems (even with \imprep{}) is $\Theta(n)$
size even for bounded-degree graphs, one might conjecture that
$\Omega(n)$ writes are required.  Our main contribution
is a (perhaps surprising) affirmative answer to the above question
for both the connectivity and biconnectivity problems.

The key technical contribution behind our breaking of the
$\Omega(n)$-write ``barrier'' is the definition and use of an
\defn{\implicit} of a graph.  Informally, a \emph{$k$-decomposition}
of a graph $G$ is comprised of a subset $S$ of the vertices, called
\emph{centers}, and a mapping $\rho(\cdot)$ that partitions the
vertices among the centers, such that (i) $|S| = O(n/k)$, (ii) the
number of vertices in each partition is at most $k$, and (iii) for
each center, the induced subgraph on its vertices is connected.
However, explicitly storing the center associated with each vertex
would require $\Omega(n)$ writes.  Instead, an \emph{\implicit}
defines the mapping implicitly in terms of a procedure that is given
only $G$ and $S$ (and a 1-bit labeling on $S$).

With the new concept of \implicit{}, we present three algorithmic
subroutines which together construct connectivity and biconnectivity
oracles with $O(m/\sqrt{\wcost})$ writes, which is $o(n)$ when $m \in
o(\sqrt{\wcost} n)$.  For clarity of presentation, we begin by
assuming the input graph has bounded degree.
Section~\ref{sec:appendix-unbounded} discusses how to relax this
constraint.

We first present an algorithm to compute an \implicit{} that can be
constructed in only $O(n/k)$ writes, $O(kn)$ reads, and low \depth,
and can compute $\rho(v)$ in only $O(k)$ expected reads and no
asymmetric memory writes.  The intuition behind our construction is
first to pick a random subset of the vertices and then map each
unpicked vertex to the closest center by performing a BFS on the graph
$G$.  Unfortunately, this does not satisfy the constraint on partition
size, so a more sophisticated approach is needed.
The unique
challenge that arises again and again in the asymmetric context is
that the sublinear limitation on the number of writes rules out the
approaches used by prior work.

We then show how the \implicit{}
can be used to solve graph connectivity and
biconnectivity.  We define the concept of a \emph{\clustergraph{}},
which contains vertices each representing a cluster and edges between
clusters.
The key idea is that after precomputing
on the \clustergraph{} and storing a constant amount of information
about connectivity and biconnectivity on each vertex (corresponding to
a cluster in the original graph), a connectivity or biconnectivity
query can be answered by only looking at the local structure and
preprocessed information on a constant number of clusters.
This is straightforward for connectivity
queries because we need only compare the labels of the clusters
that contains the respective query points.
However, this becomes much more challenging in graph biconnectivity
since the correspondence between the clusters and biconnected components is non-trivial:
a cluster may contain the vertices in many biconnected components while the vertices in a certain biconnected component can belong to different clusters.
Therefore, biconnectivity queries
require considerable subtleties in the design, to store the appropriate
information on the \clustergraph{} that enables each query to
access only a constant number of clusters (at most 3).
More specifically, we define the concept of the \emph{\localgraph{}} of each cluster (it maintains the relationship of biconnectivity of this cluster and its neighbor clusters and can be computed with the cost proportional to the size of this cluster), such that the biconnectivity queries discussed in Section~\ref{sec:biconn} can be answered by looking up a constant number of \localgraph{}s and the information stored in the \clustergraph{}.

Our sequential algorithms have significant algorithmic merits on their own, but we also show that all the algorithms can be
made to run in parallel with low \depth.  We show:



\begin{theorem}
\label{thm:main-sublinear}
Graph connectivity and biconnectivity oracles can be constructed in
$O(m/\sqrt{\wcost})$ expected writes and
$O(m\sqrt{\wcost})$ expected time (parallel work) on the \seqmodel{}
model (\ourmodel{} model, respectively).
The \depth{} on the \ourmodel{} is $O(\wcost^{3/2}\log^3 n)$ \whp{}.
Each connectivity query can be answered in $O(\wcost^{1/2})$ expected time
(work) ($O(\wcost{}^{1/2}\log n)$ \whp{}) and each biconnectivity query in
$O(\wcost)$ expected time (work) ($O(\wcost{}\log n)$ \whp{}).
The \local{} memory is $O(\wcost \log n)$ words.
\end{theorem}

\hide{
\medskip\noindent
{\bf Do connectivity/biconnectivity require $\Omega(n)$ writes?}
Note that for sparse graphs, the work of our connectivity and biconnectivity algorithms
is dominated by the $\Theta(n)$ writes it performs (and prior algorithms are worse).
This led us to explore the following fundamental question:
\begin{quote}
\emph{Do connectivity and biconnectivity require $\Omega(n)$ writes?}
\end{quote}
Note that the answer to this question is uncertain even for bounded-degree graphs,
because the standard output for these problems is $\Omega(n)$ size.
Thus, a more precise phrasing of the question would be:
\emph{Is it possible to construct, using $o(n)$ writes
to the asymmetric memory, an oracle for connectivity (for biconnectivity)
in a bounded-degree graph that can answer queries in time independent of $n$?}

\myparagraph{Breaking the $\Omega(n)$ write barrier for bounded-degree graphs}
We answer affirmative to this latter question for both the
connectivity and biconnectivity problems.
The key idea is the
definition and use of an \defn{\implicit} of a graph.  Informally, a
\emph{$k$-decomposition} of a graph $G$ is comprised of a subset $S$
of the vertices, called \emph{centers}, and a mapping $\rho(\cdot)$ that
partitions the vertices among the centers, such that (i) $|S| =
O(n/k)$, (ii) the number of vertices in each partition is at most $k$,
and (iii) for each center, the induced subgraph on its vertices is
connected.  However, explicitly storing the center associated with
each vertex would require $\Omega(n)$ writes.  Instead, an
\emph{\implicit} defines the mapping implicitly in terms of a
procedure that is given only $G$ and $S$ (and a 1-bit labeling on
$S$).  The advantage of an \implicit{} is that storing $S$ requires
only $O(n/k)$ writes.  On the other hand, it is not clear that there
exists such a decomposition, with an efficient construction and fast
computation of $\rho(v)$ given $v$, because of the requirement that
the induced subgraphs be connected.

In this paper, we present an \implicit{} for bounded-degree graphs
that can be constructed in only $O(n/k)$ writes, $O(kn)$ reads, and
low \depth, and can compute $\rho(v)$ in only $O(k)$ reads and no
writes.  The intuition behind our construction is first to pick a
random subset of the vertices and then map each unpicked vertex to the
closest center by performing a BFS on the graph $G$.  Unfortunately,
this does not satisfy the constraint on partition size, so a more
sophisticated approach is needed.  The unique challenge that arises
again and again in the asymmetric context is
that the sublinear limitation on the number of writes rules out
the approaches used by prior work.

Although graph decompositions with various properties have been shown
to be quite useful in a large variety of applications (e.g.,
\cite{awerbuch1985complexity, miller2013parallel, blelloch2014nearly,
  awerbuch1992low, awerbuch1989network, abraham2012using,
  linial1991decomposing}), to our knowledge none of the prior
algorithms provide the necessary conditions to process graphs with a
sublinear number of writes.  For example, Miller et
al.~\cite{miller2013parallel} developed a simple and efficient
parallel \emph{low-diameter} decomposition algorithm that takes a
parameter $\beta <1/2$ and builds a decomposition with clusters of
diameter $O(\log (n)/ \beta)$ and $\beta m$ crossing edges with high
probability, in expected $O((\log^2 n)/{\beta})$ \depth{} and $O(m)$
work. However, even if the algorithm's many breadth-first-searches are
executed using a write-efficient BFS~\cite{BBFGGMS16}, it still
requires at least $O(n)$ writes. Furthermore, there is no guarantee on
the partition sizes.  Our decomposition needs only to be connected and
small, whereas other decompositions may require low-diameter,
relatively few edges between clusters, etc.  Further work has
considered balanced graph partitioning \cite{andreev2004balanced},
where, given parameters $k, u$, the goal is to find a partition of the
vertices into $k$ clusters such that each cluster is of size at most
$u \cdot n/k$.

\myparagraph{Contributions}
In summary, our main contributions are as follows:
\begin{itemize}
\item
We define an \implicit{} of a graph and show how it can be constructed efficiently
(and in parallel) for bounded-degree graphs using only $O(n/k)$ writes to
the asymmetric memory.
\item
We show how an \implicit{} can be used to break the $\Omega(n)$ barrier for connectivity
(biconnectivity) on bounded-degree graphs, resulting in $O(n \sqrt{\wcost})$ expected work, $O(n/\sqrt{\wcost})$
writes, low \depth{} algorithms for constructing a connectivity (biconnectivity) oracle
that answers queries in an expected $O(\sqrt{\wcost})$ work ($O(\wcost)$ work, respectively).
These algorithms are interesting in their own right, as there are considerable
subtleties in their design, particularly in the biconnectivity case.
\item
For general graphs, we present the first $O(m)$ work, low \depth{} parallel algorithms
for connectivity and for biconnectivity in the asymmetric memory setting.
\end{itemize}
Our algorithms support related connectivity problems such as spanning forests, articulation
points, bridges, cut-block trees, and 1-edge connectivity queries.

}

\section{Preliminaries and Related Work}

Let $G=(V, E)$ be an undirected, unweighted graph with $n = |V|$ vertices
and $m = |E|$ edges.
$G$ can contain self-loops and parallel (duplicate) edges, and is not necessarily connected.
We assume a global ordering of the vertices to break ties when necessary.
If the degree of every vertex is bounded by a constant, we say the graph
is \defn{bounded-degree}.
We use standard definitions of \emph{spanning tree}, \emph{spanning
  forest}, \emph{connected component}, \emph{biconnected component},
\emph{articulation points}, \emph{bridge}, and
\emph{$k$-edge-connectivity} on a graph, and
\emph{lowest-common-ancestor} (LCA) query on a tree (as summarized in
Appendix~\ref{sec:appendix-prelim}).
Let $[n] = \{1, 2, \cdots, n\}$ where $n$ is a positive integer.


\myparagraph{Computation models}
Sequential algorithms are analyzed using the \defn{\seqmodel}
model~\cite{blelloch2016efficient}, comprised of an infinitely large
asymmetric memory and a small \local{} memory.  The cost of writing to
the large memory is $\wcost$, and all other operations have unit cost.
This models practical settings in which there is a small amount of
standard symmetric memory (e.g., a cache) in addition to the asymmetric memory.

For parallel algorithms, we use the \defn{\ourmodellong\ (NP)} 
model~\cite{BBFGGMS16}, which is designed to characterize both parallelism
and memory read-write asymmetry. In the model, a
computation is represented as a (dynamically unfolding) directed acyclic graph (DAG) of tasks that
begins and ends with a single task called the root. A task consists of
a sequence of instructions that must be executed in order. Tasks can
also call the Fork instruction, which creates child tasks that can be
run in parallel with each other. The memory in the \ourmodel\ Model
consists of (i) an infinitely large
\emph{asymmetric} memory accessible to all tasks and (ii) a small
task-specific \emph{\local} memory accessible only to a task
and its children.  The cost of writing to large memory is
$\wcost$, and all other operations have unit cost.
%
The \defn{work} $W$ of a computation is the sum of the costs of the
operations in its DAG and the \defn{\depth{}} $D$ is the
cost of the DAG's most expensive path.  Under mild assumptions,
Ben-David et al.~\cite{BBFGGMS16} showed that a work-stealing scheduler
can execute an algorithm whose \ourmodel{} complexity is
work $W$ and \depth{} $D$ in $O(W / P + \wcost D)$ expected time
on $P$ processors.

In both models,
the number of \defn{writes} refers only to the writes to the
asymmetric memory, ignoring any writes to \local{}
memory.  All reads and writes are to words of size $\Theta(\log n)$
for input size $n$.  The size of the \local{} memory is measured in
words.

\myparagraph{Related Work}
Read-write asymmetries have been studied in the context of NAND Flash
chips~\cite{BT06, Eppstein14, Gal05, ParkS09}, focusing on how to
balance the writes across the chip to avoid uneven wear-out of
locations.  Targeting instead the new memory technologies, read-write
asymmetries have been an active area of research in the
systems/database/architecture communities (e.g., \cite{Arulraj17sigmod,
  Bausch2012damon, Chen11, Chen2015pvldb, ChoL09, LeeIMB09,
  Oukid2016sigmod, Viglas12, Viglas14, wang2013wade, yang:iscas07,
  zhou2012writeback, ZhouZYZ09, ZWT13}).  In the algorithms community,
Blelloch et al.~\cite{BFGGS15} defined several sequential and parallel
computation models that take asymmetric read-write costs into account,
and analyzed and designed sorting algorithms under these models.
Their follow-up paper~\cite{blelloch2016efficient} presented
sequential algorithms for various problems that do better than their
classic counterparts under asymmetric read-write costs, as well as
several lower bounds.  Carson et al.~\cite{carson2016write} presented
write-efficient sequential algorithms for a similar model, as well as
write-efficient parallel algorithms (and lower bounds) on a
distributed memory model with asymmetric read-write costs, focusing on
linear algebra problems and direct N-body methods.  Ben-David et
al.~\cite{BBFGGMS16} proposed a nested-parallel model with asymmetric
read-write costs and presented write-efficient, work-efficient, low
\depth{} (span) parallel algorithms for reduce, list contraction, tree
contraction, breadth-first search, ordered filter, and planar convex
hull, as well as a write-efficient, low-\depth{} minimum spanning tree
algorithm that is nearly work-efficient.  Jacob and
Sitchinava~\cite{jacob2017} showed lower bounds for an asymmetric
external memory model.  In each of these models, there is a small
amount of \local{} memory that can be used to help minimize the number
of writes to the large asymmetric memory.

Although graph decompositions with various properties have been shown
to be quite useful in a large variety of applications (e.g.,
\cite{abraham2012using, awerbuch1985complexity, awerbuch1992low, awerbuch1989network, blelloch2014nearly, linial1991decomposing, miller2013parallel}),
to our knowledge none of the prior
algorithms provide the necessary conditions for processing graphs
with a sublinear number of writes in order to answer connectivity/biconnectivity
queries (targeting instead other decomposition properties
that are unnecessary in our setting, such as few edges between
clusters).  For example, Miller et al.'s~\cite{miller2013parallel}
parallel low-diameter decomposition algorithm requires at least
$\Omega(n)$ writes (even if a write-efficient BFS~\cite{BBFGGMS16} is
used), and provides no guarantees on the partition sizes.  Similarly,
algorithms for size-balanced graph partitioning (e.g.,
\cite{andreev2004balanced}) require $\Omega(n)$ writes.  Our
\implicit{} construction is reminiscient of sublinear time algorithms
for estimating the number of connected
components~\cite{Berenbrink2014ipl, chazelle2005siam} in its use of
BFS from a sample of the vertices.  However, their BFS is used for a
completely different purpose (approximate counting of $1/n_u$, the
inverse of the size of the connected component containing a sampled
node $u$), does not provide a partitioning of the nodes into clusters
(two BFS from different nodes can overlap), and cannot be used for
connectivity or biconnectivity queries (two BFS from the same
connected component may be truncated before intersecting).

\section{Implicit Decomposition}\label{sec:implicit}

In this paper we introduce the concept of an implicit decomposition.
The idea is to partition a graph into connected clusters such that all we need to store to
represent the cluster is one representative, which we call the
center of the cluster, and some small amount of information on that
center (1 bit in our case).  The goal is to
quickly answer queries on the cluster.  The queries we consider are:
given a vertex find its center, and given a center
find all its vertices.  To reduce the amount of \local-memory
needed, we need all clusters to be roughly the same size.
We start with some definitions, which consider only undirected
graphs.

For graph $G = (V,E)$ we refer to the subgraph induced by a subset of
vertices as a \defn{cluster}.  A \defn{decomposition} of a connected
graph $G=(V,E)$ is a vertex subset $S\subset V$, called
\defn{centers}, and a function $\rho(v):V\to S$, such that the
\defn{cluster} $C(s) = \{v \in V~|~\rho(v) = s\}$ for each center $s \in S$
is connected.  A decomposition is a \defn{k-decomposition} if the
size of each cluster is upper bounded by $k$, and $|S|=O(n/k)$
(i.e. all clusters are about the same size).
We are often interested in the graph induced by the
decomposition, and in particular:
\begin{definition}[\clustergraph]
Given the decomposition $(S,\rho)$ of a graph $G = (V,E)$, the \defn{\clustergraph} is the multigraph $G' = (S,\langle\ \{\rho(u),
    \rho(v)\} : \{u,v\} \in E, \rho(u) \neq \rho(v)\ \rangle\ )$.
\end{definition}

\begin{definition}[implicit decomposition]
An \defn{implicit decomposition} of a connected graph
$G=(V,E)$ is a decomposition $(S,\rho,\ell)$ such that $\rho(\cdot)$ is
defined implicitly in terms of an algorithm given only $G$, $S$, and a
(short) labeling $\ell(s)$ on $s \in S$.
\end{definition}

\begin{figure}[t]
\centering
  \includegraphics[width=.8\columnwidth]{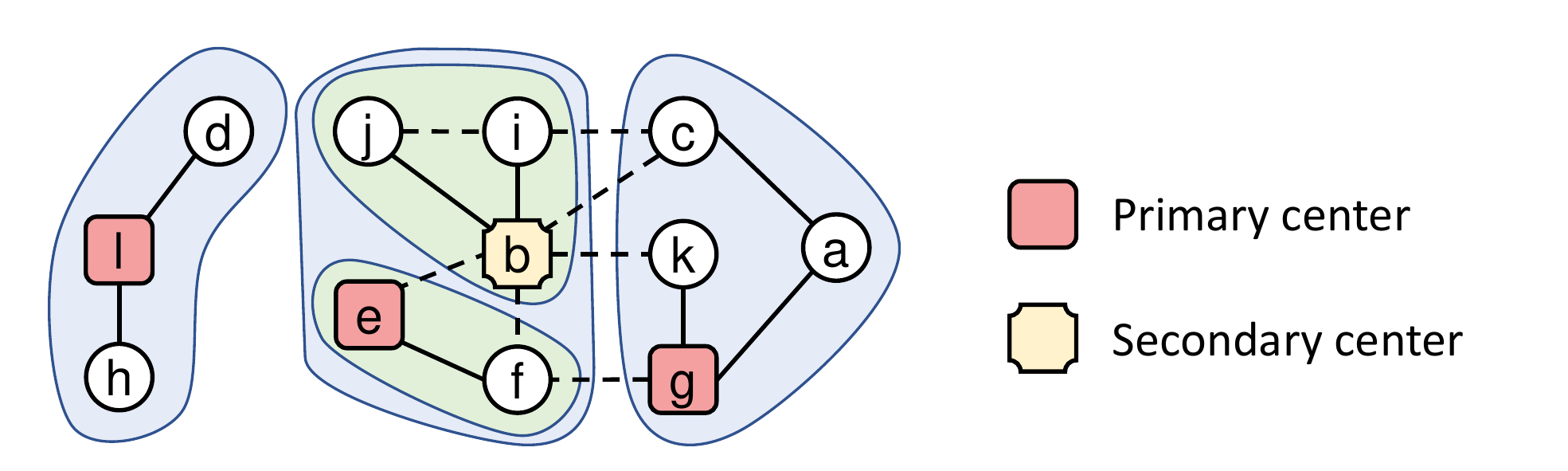}
  \vspace{.05in}
\caption{An example \implicit{} for $k = 4$ consisting of clusters $\{d,h,l\}, \{i,j,b\},\{e,f\},$ and $\{a,c,g,k\}$  .  In the graph,
  $j's$ primary center is $e$ (i.e., $\rho_0(j) = e$) and its
  secondary center is $b$ (i.e., $\rho(j) = b$).   Note $b$ is on a
  shortest path to $e$.   Also note that $c$ is closer to the
  secondary center $b$ than to $g$, but picks $g$ as its secondary
  (and primary) center, since $g$ is its primary center.  In breaking
  ties we assume lower letters have higher priorities.}
\label{fig:cluster}
 \end{figure}

In this paper, we use \implicit{}s.  An example is given in
Figure~\ref{fig:cluster}.  Our goal is to construct and query the
decomposition quickly, while using short labels.  Our main result is
the following.

\begin{theorem}\label{thm:mainimp}
An \implicit{} $(S,\rho,\ell)$ can be constructed on an
undirected bounded-degree graph $G = (V,E)$ with $|V| = n$ such that:
\begin{itemize}
\setlength{\itemsep}{1pt}
\setlength{\parskip}{0pt}
\setlength{\parsep}{0pt}
\item the construction takes
$O(kn)$ operations and $O(n/k)$ writes, both in expectation;
\item  $\rho(v), v \in V$ takes $O(k)$ operations in expectation
and $O(k \log n)$ \whp{}, and no writes;
\item $C(s), s \in S$ takes $O(k^2)$ operations in expectation,
  and $O(k^2 \log n)$ \whp{} and no writes;
\item the labels $\ell(s), s \in S$ are $1$-bit each; and,
\item construction and queries take $O(k \log n)$ \local{}
  memory \whp{}.
\end{itemize}
\end{theorem}
\noindent
Note that this theorem indicates a linear tradeoff of reads
(operations) and writes for the construction that is controlled by
$k$.

At a high level, the construction algorithm works by identifying a
subset of centers such that every vertex can quickly find its nearest
center without having to keep a pointer to it (which would require too
many writes).  It first selects a random subset of the vertices where
each vertex is selected with probability $1/k$.  We call these the
\defn{primary centers} and denote them as $S_0$.  All other vertices
are then assigned to the nearest such center.  Unfortunately, a
cluster defined in this way can be significantly larger than $k$
(super polynomial in $k$).  To handle this, the algorithm identifies an
additional $O(n/k)$ \defn{secondary centers}, $S_1$.  Every vertex $v$
is associated with a primary center $\rho_0(v) \in S_0$, and an actual
center $\rho(v) \in S = S_0 \cup S_1$.  The only values the the
algorithm stores are the set $S$ and the bits $\ell(s), s \in S$
indicating whether it is primary or secondary.

\newcommand{\p}{\mb{SP}}

In our construction it is important to break ties among equal-length
paths in a consistent way, such that subpaths of a shortest path are
themselves a unique shortest path.  For this purpose we assume the
vertices have a total ordering (and comparing two vertices takes
constant time).  Among two equal hop-length paths from a vertex $u$,
consider the first vertex where the paths diverge.  We say that the
path with the higher priority vertex at that position is shorter.  Let
$\p(u,v)$ be the shortest path between $u$ and $v$ under this
definition for breaking ties, and $L(\p(u,v))$ be its length such that
comparing $L(\p(u,v))$ and $L(\p(u,w))$ breaks ties as defined.  By
our definition all subpaths of a shortest path are also unique
shortest paths for a fixed vertex ordering.  Based on these
definitions we specify $\rho_0(v)$ and $\rho$ as follows:

$$\rho_0(v) = \argmin_{u \in S_0} L(v,u)$$
$$\rho(v) = \argmin_{u \in S \wedge u \in {\scriptsize \p}(v,\rho_0(v))} L(v,u)$$

The definitions indicate that a vertex's center is the first center
encountered on the shortest path to the nearest primary center.  This
could either be a primary or secondary center (see
Figure~\ref{fig:cluster}). $\rho(v)$ is defined in this manner to
prevent vertices from being reassigned to secondary centers created in
other primary clusters, which could result in oversized clusters.

We now describe how to find $\rho(v)$ for every vertex $v$.
First, we find $v$'s closest primary center by running a
BFS from $v$ until we hit a vertex in $S_0$. The BFS orders the vertices by $L(\p(v,\cdot))$.
To find $\rho(v)$
we first search for the primary center of $v$ ($\rho_0(v)$) and then
identify the first center on the path from $v$ to $\rho_0(v)$,
possibly $\rho_0(v)$ itself.

\begin{lemma}
\label{lemma:findcenter}
$\rho(v)$ can be found in $O(k)$ operations in expectation, and
$O(k \log n)$ operations \whp{}, and using
$O(k \log n)$ \local{} memory \whp{}.
\end{lemma}
\begin{proof}
Note that the search order from a vertex is deterministic and
independent of the
sampling used to select $S_0$.
Therefore, the expected number of vertices visited before hitting a
vertex in $S_0$ is $k$.  By tail bounds, the probability
of visiting $O(c k \log n)$ vertices before hitting one in $S_0$ is
at most $1/n^c$.  The search is a BFS, so it takes time linear
in the number of vertices visited.  Since the vertices are of bounded
degree, placing them in priority order in the queue is easy.  Once the
primary center is found, a search back on the path gives the actual center.
We assume that \local{} memory is
used for the search so no writes to the asymmetric memory are
required.  The memory used is proportional to the search size, which
is proportional to the number of operations; $O(k)$ in expectation and $O(k
\log n)$ \whp{}.
\end{proof}

The space requirement for the \local{} memory is $O(k \log n)$, which is considered to be realistic and easy to satisfied since we set $k=\sqrt{\wcost{}}$ when using this decomposition later in this paper.

We use the following lemma\hide{, whose proof is in
Appendix~\ref{sec:appendix-implicit},} to help find $C(s)$ for a center $s$.

\begin{lemma}
	\label{lemma:clustertree}
	The union of the shortest paths $\p(v,\rho(v))$ for $v \in V$
        define a rooted spanning tree on each cluster, with the center
        as the root (assuming the path edges are directed to $\rho(v)$).
\end{lemma}

\begin{proof}
	We first show that this is true for the clusters defined by the
	primary centers $S$ ($\rho_0(v))$.
        We use the notation
        $\p(u,v)+\p(v,w)$ to indicate joining the two shortest paths
        at $v$. Consider a vertex $v$ with
        Consider a vertex $v$ with
	$\rho_0(v) = s$, and consider all the vertices $P$ on the shortest
	path from $v$ to $s$.  The claim is that for each $u \in P, \rho(u) =
	s$ and $\mb{SP}(u,s)$ is a subpath of $P$. This implies a
	rooted tree. To see that $\rho(u) = s$ note that the shortest path
	from $u$ to a primary vertex $t$ has length $L(\mb{SP}(u,t))$. We can write the
	length of the shortest path from $v$ to $t$ as
	$L(\mb{SP}(v,t)) \leq L(\mb{SP}(v,u) + \mb{SP}(u,t))$ and the length of the
	shortest path from $v$ to $s$ as
	$L(\mb{SP}(v,s)) = L(\mb{SP}(v,u) + \mb{SP}(u,s))$.
	We know that since $\rho_0(v) = s$
	that $L(\mb{SP}(v,s)) < L(\mb{SP}(v,t))$ $\forall t \neq s$. Through substitution and subtraction,
	we see that $L(\mb{SP}(u,s)) < L(\mb{SP}(u,t))$ $\forall t \neq s$. This means that $\rho_0(u) = s$.
	We know that $\mb{SP}(u,s)$ cannot contain the edge $b$ that $v$ takes to reach $u$
	in $\mb{SP}(v,s)$ since $u \in \mb{SP}(v,s)$. Since the search from $u$ excluding
	$b$ will have the same priorities as the search from $v$ when it
	reaches $u$, $\mb{SP}(u,s)$ is a subpath of $P$.
	
	Now consider the clusters defined by $\rho(v)$.  The secondary centers
	associated with a primary center $s$ partition the tree for $s$ into
	subtrees, each rooted at one of those centers and going down until
	another center is hit.  Each vertex in the tree for $s$ will be
	assigned the correct partition by $\rho(v)$ since each will be
	assigned to the first secondary center on the way to the primary
	center.
\end{proof}

The set of solid edges in Figure~\ref{fig:cluster} is an example of
the spanning forest.  This gives the following.

\begin{corollary}
	\label{cor:clustertree}
For any vertex $v$, $\p(v,\rho(v)) \subseteq C(\rho(v))$.
\end{corollary}

\begin{lemma}
\label{lemma:findcluster}
For any vertex $s \in S$, its cluster $C(s)$
can be found in $O(k \lvert C(s)\rvert)$ operations in expectation and $O(k
\lvert C(s)\rvert \log n)$ operations \whp{}, and using $O(\lvert C(v)\rvert + k\log n)$ \local{} memory \whp{}.
\end{lemma}
\begin{proof}
For any center $s \in S$, identifying all the vertices in its cluster
$C(s)$ can be implemented as a BFS starting at $s$.  For each vertex
$v \in V$ that the BFS visits, the algorithm checks if $\rho(v) = s$.
If so, we add $v$ to $C(s)$ and put its unvisited neighbors in the BFS
queue, otherwise we do neither.
By Corollary~\ref{cor:clustertree}, any
vertex $v$ for which $\rho(v) = s$ must have a path to $s$ only through
other vertices whose center is $v$.   Therefore the algorithm will
visit all vertices in $C(s)$.  Furthermore, since the graph has
bounded degree it will only visit $O(C(s))$ vertices not in $C(s)$.
Each visit to a vertex $u$ requires finding $\rho(v)$.  Our bound on the
number of operations therefore follows from Lemma~\ref{lemma:findcenter}.
We use $O(|C(v)|)$ \local{} memory for storing the queue and $C(v)$,
and $O(k\log n)$ \whp{} for calculating $\rho(v)$.
\end{proof}

\begin{algorithm}[t]
\caption{Constructing $k$-Implicit Decomposition}
\label{algo:genclusters}

\KwIn{Connected bounded-degree graph $G=(V,E)$, parameter $k$}
\KwOut{A set of cluster centers $S_0$ and $S_1$ ($S = S_0 \bigcup S_1$)}
\smallskip
    Sample each vertex with probability $1/k$, and place in $S_0$\label{line:sample}\\
    $S_1 = \varnothing$\\
    \ForEach {\upshape vertex $v\in S_0$} {
       \mf{SecondaryCenters}($v$, $G$, $S_0$)
    }
    \Return $S_0$ and $S_1$\\
    \smallskip
    \SetKwProg{myfunc}{function}{}{}
    \myfunc{\upshape \mf{SecondaryCenters}($v$, $G$, $S$)} {
       Search from $v$ for the first $k$ vertices
       that have $v$ as their center.  This defines a tree.\\
       If the search exhausts all vertices with center $v$, \Return.\\
       Otherwise identify a vertex $u$ that partitions the tree
       such that its subtree and the rest of the tree
       are each at least a constant fraction of $k$.\\
       Add $u$ to $S_1$.\\
       \mf{SecondaryCenters}($v$, $G$, $S \cup u$)\\
       \mf{SecondaryCenters}($u$, $G$, $S \cup u$)
   }
\end{algorithm}

We now show how to select the secondary centers such that
the size of each cluster is at most $k$.
Algorithm~\ref{algo:genclusters} describes the process.  By
Lemma~\ref{lemma:clustertree}, before the secondary centers are
added, each primary vertex in $s \in S_0$ defines a rooted tree of
paths from the vertices in its cluster to $s$.  The function
\mf{SecondaryCenters} then recursively cuts up this tree into subtrees
rooted at each $u$ that is identified.

\begin{lemma}\label{lemma:genclusters}
Algorithm~\ref{algo:genclusters} runs in
$O(n k)$ operations and $O(n/k)$ writes (both in expectation), and
$O(k \log n)$ \local{} memory \whp{} on the \seqmodel{} Model.
It generates a $k$-implicit
decomposition $S$ of $G$ with labels distinguishing $S_0$ from $S_1$.
\end{lemma}

\begin{proof}
The algorithm creates clusters of size at most $k$ by construction (it
partitions any cluster bigger than $k$ using the added vertices $u$).
Each call to \mf{SecondaryCenters} (without recursive calls) will use
$O(k^2)$ operations in expectation since we visit $k$ vertices and each one
has to search back to $v$ to determine if $v$ is its center.  Each
call also uses $O(k \log n)$ space for the search \whp{} since we need
to store the $k$ elements found so far and each $\rho(v)$ uses $O(k
\log n)$ space for the search \whp{}.  Before making the recursive
calls, we can throw out the \local{} memory and write out $u$ to $S_1$,
requiring one write per call to \mf{SecondaryCenters}.

We are left with showing there are at most $O(n/k)$ calls to
\mf{SecondaryCenters}.  There are $n/k$ primary clusters in
expectation.  If there are too many (beyond some constant factor above
the expectation), we can try again.  Since the graph has bounded
degree, we can find a vertex that partitions
the tree such that its subtree and the rest of the tree are both at
most a constant fraction~\cite{rosenberg2001graph}.  We can now count all internal nodes
of the recursion against the leaves.  There are at most $O(n/k)$
leaves since each defines a cluster of size $\Theta(k)$.  Therefore
there are $O(n/k)$ calls to \mf{SecondaryCenters}, giving the overall
bounds stated.
\end{proof}

\myparagraph{Parallelizing the decomposition}
To parallelize the decomposition in Algorithm~\ref{algo:genclusters},
we make one small change; in addition to adding the
secondary cluster $u$ at each recursive call to \mf{SecondaryCenters},
we add all children of $v$.  This guarantees that the height of the tree
decreases by at least one on each recursive call, and only increases
the number of writes by a constant factor.  This gives the following
lemma. 

\begin{lemma}\label{lemma:pargenclusters}
Algorithm~\ref{algo:genclusters} runs in
depth $O((k \log n) (k^2 \log n + \omega))$ \whp{} on the
\ourmodel.
\end{lemma}
\begin{proof}
	Certainly selecting the set $S_0$ can be done in parallel.
	Furthermore the calls to \mf{SecondaryCenters} on line 4 can be
	made recursively in parallel.  The depth will be proportional to the
	depth to each call to \mf{SecondaryCenters} (not including
	recursive calls) multiplied by the depth of the recursion.
	To bound the depth, in the parallel version we also mark the children of the root as secondary centers, which does not increase the number of secondary centers asymptotically (due to the bounded-degree assumption).
	In this way one is removed from the height of the tree on each recursive call.
	The depth of the recursion is at most the depth of the tree
	associated with the primary center $\rho_0(v)$.  This is bounded by
	$O(k \log n)$ \whp{} since by Lemma~\ref{lemma:findcenter} every
	vertex finds its primary center within $O(k \log n)$ steps \whp{}.
	The depth of \mf{SecondaryCenters} (without recursion) is just the
	number of operations ($O(k^2 \log n)$ \whp{}) plus the depth of the one
	write of $u$ (which costs $\omega$).  This gives the bound.
\end{proof}

\myparagraph{Extension to unconnected graphs} Note that once a
connected component contains at least one primary center, the
definition and Theorem~\ref{thm:mainimp} hold.  However, it is
possible that in a small component, the search of $\rho(\cdot)$
exhausts all connected vertices without finding any primary centers
(vertices in the initial sample, $S_0$).  In this case, we check
whether the size of the cluster is at least $k$, and if so, we mark as
a primary center the vertex that is the smallest according to the
total order on vertices.
This marks at most $n/k$ primary centers and the rest of the algorithm
remains unchanged.  This step is added after line~\ref{line:sample} in
Algorithm~\ref{algo:genclusters}, and requires $O(nk)$ work and
operations, $O(n/k)$ writes, and $O(k)$ \depth.  The cost bound
therefore is not changed.  If the component is smaller than $k$, we
use the smallest vertex in the component as a center implicitly, but
never write it out.  The $\rho(\cdot)$ function can easily return this
in $O(k)$ operations.

\hide{
\myparagraph{Extension to unconnected graphs}
In the above discussion, we assumed the input graph is connected.
However, for some problems, like graph connectivity, the graph is not necessarily connected.
In Appendix~\ref{sec:appendix-implicit}, we show how to extend the definition of \implicit{} and the algorithm to generate the cluster centers for unconnected input graphs.
}

\section{Graph Connectivity and Spanning Forest}\label{sec:cc}

This section describes parallel write-efficient algorithms for graph
connectivity and spanning forest; that is, identifying which vertices
belong to each connected component and producing a spanning forest of
the graph.  These task can be easily accomplished sequentially by
performing a breadth-first or depth-first search in the graph with
$O(m)$ operations and $O(n)$ writes.  While there are several
work-efficient parallel algorithms for the
problem~\cite{Shun2014,ColeKT96,Gazit1991,PoonR97,PettieR02,Halperin00,HalperinZ96},
all of them use $\Omega(n+m)$ writes.  This section has two main
contributions: (1) Section~\ref{sec:cc-linear} provides a parallel
algorithm using $O(n+m/\wcost)$ writes in expectation,
$O(n\wcost + m)$ expected work, and $O(\wcost^2\log^2n)$ \depth{} with
high probability; (2) Section~\ref{sec:cc-sublinear} gives an
algorithm for constructing a connectivity oracle on constant-degree
graphs in $O(n/\sqrt{\wcost})$ expected writes and $O(n\sqrt{\wcost})$ expected
total operations.  Our oracle-construction algorithm is parallel,
having \depth{} $O(\wcost^{3/2}\log^3n)$ \whp{}, but it also
represents a contribution as a sequential algorithm.

Our parallel algorithm (Section~\ref{sec:cc-linear}) can be viewed as
a write-efficient version of the parallel algorithm due to Shun et
al.~\cite{Shun2014}.  This algorithm uses a low-diameter decomposition
algorithm of Miller et al.~\cite{miller2013parallel} as a subroutine,
which we review and adapt next in Section~\ref{sec:ldd} and
Appendix~\ref{sec:appendix-cc}.


\subsection{Low-diameter Decomposition}\label{sec:ldd}
Here we review the low-diameter decomposition of Miller et
al.~\cite{miller2013parallel}.  The so-called
``$(\beta,d)$-decomposition'' is terminology lifted from their paper,
and it should not be confused with our implicit
$k$-decompositions. The details of the decomposition subroutine are
only important to extract a bound on the number of writes, and it is
briefly summarized in Appendix~\ref{sec:appendix-cc}.

A \defn{$(\beta,d)$-decomposition} of an undirected graph $G = (V,E)$,
where $0<\beta<1$ and $1\leq d \leq n$, is defined as a partition of
$V$ into subsets $V_1,\ldots,V_k$ such that (1) the shortest path
between any two vertices in each $V_i$ using only vertices in $V_i$ is
at most $d$, and (2) the number of edges $(u,v) \in E$ crossing the
partition, i.e., such that $u \in V_i$, $v \in V_j$, and $i \neq j$,
is at most $\beta m$.
Miller et al.~\cite{miller2013parallel} provide an efficient parallel
algorithm for generating a $(\beta,O(\log n / \beta))$-decomposition.
As described, however, their algorithm performs $\Theta(m)$ writes.
The key subroutine of the algorithm, however, is just breadth-first
searches (BFS's).  Replacing these BFS's by write-efficient
BFS's~\cite{BBFGGMS16} yields the following theorem:

\begin{theorem}\label{thm:ldd}
A $(\beta, O({\log n}/\beta))$ decomposition can be generated in
$O(n)$ expected writes, $O(m+\wcost n)$ expected work, and
$O(\wcost{\log^2 n}/\beta)$ \depth\ \whp{} on the \ourmodel\ model.
\end{theorem}

\subsection{Connectivity and Spanning Forest}\label{sec:cc-linear}
The parallel connectivity algorithm of~\cite{Shun2014} applies the
low-diameter decomposition recursively with $\beta$ set to a
constant less than $1$. Each level of recursion contracts a subset of vertices into
a single supervertex for the next level. The algorithm terminates when
each connected component is reduced to a single supervertex.  The
stumbling block for write efficiency is this contraction step, which
performs writes proportional to the number of remaining edges.

Instead, our write-efficient algorithm applies the low-diameter
decomposition just once, but with a much smaller $\beta$, as follows:
\begin{enumerate}
\item Perform the low-diameter decomposition with parameter
  $\beta = 1/\wcost$.
\item Find a spanning tree on each $V_i$ (in parallel) using
  write-efficient BFS's of~\cite{BBFGGMS16}.
\item Create a contracted graph, where each vertex subset in the
  decomposition is contracted down to a single vertex.  To write down
  the cross-subset edges in a compacted array, employ the
  write-efficient filter of~\cite{BBFGGMS16}.
\item Run any parallel linear-work spanning forest algorithm on the
  contracted graph, e.g., the algorithm from~\cite{ColeKT96} with
  $O(\wcost \log n)$ \depth{}.
\end{enumerate}
Combining the spanning forest edges across subsets (produced in Step
4) with the spanning tree edges (produced in Step 2) gives a spanning
forest on the original graph.  Adding the bounds for each step
together yields the following theorem.  Again only $O(1)$ \local{}
memory is required per task.

\begin{theorem}\label{thm:cc-linear}
  For any choice of $0 < \beta < 1$, connectivity and spanning forest
  can be solved in $O(n+\beta m)$ expected writes,
  $O(\wcost n + \beta \wcost m + m)$ expected work, and
  $O(\wcost \log^2 n / \beta)$ \depth{} \whp{} on the \ourmodel{}
  model.  For $\beta = 1/\wcost$, these bounds reduce to
  $O(n + m/\wcost)$ expected writes, $O(m+\wcost n)$ expected work and
  $O(\wcost^2\log^2n)$ \depth{} \whp{}.
\end{theorem}
\begin{proof}
  Step~1 has performance bounds given by Theorem~\ref{thm:ldd}, and
  the expected number of edges remaining in the contracted graph is at
  most $\beta m$.  Step~2 performs BFS's on disjoint subgraphs, so
  summing across subsets yields $O(n)$ expected writes and
  $O(m+n\wcost)$ expected work.  Since each tree has low diameter
  ${\cal D} = O(\log n / \beta)$, the BFS's have \depth{} $O(\wcost
  {\cal D} \log n) = O(\wcost \log^2 n / \beta)$
  \whp{}~\cite{BBFGGMS16}.  Step~3 is dominated by the filter, which
  has a number of writes proportional to the output size of $O(\beta
  m)$, for $O(m + \beta \wcost m)$ work. The \depth{} is $O(\wcost
  \log n)$~\cite{BBFGGMS16}.  Finally, the algorithm used in Step~4 is
  not write-efficient, but the size of the graph is $O(n+\beta m)$,
  giving that many writes and $O(\wcost (n + m\beta))$ work.  Adding
  these bounds together yields the theorem.
\end{proof}

\subsection{Connectivity Oracle in Sublinear Writes}\label{sec:cc-sublinear}
A connectivity oracle supports queries that take as input a vertex and
return the label (component ID) of the vertex.  This allows one to
determine whether two vertices belong in the same component.  The
algorithm is parameterized by a value $k$, to be chosen later.  We
assume throughout that the \local{} memory per task is
$\Omega(k\log n)$ words and that the graph has bounded degree.

We begin with an outline of the algorithm.  The goal is to produce an
oracle that can answer for any vertex which component it belongs to in
$O(k)$ work.  To build the oracle, we would like to run the
connectivity algorithm on the \clustergraph{} produced by an implicit
$k$-decomposition. The result would be that all center vertices in the
same component be labeled with the same identifier.  Answering a query
then amounts to outputting the component ID of the center it maps to,
which can be queried in $O(k)$ expected work and $O(k\log n)$ work
\whp{} according to Lemma~\ref{lemma:findcenter}.

The main challenge in implementing this strategy is that we cannot
afford to write out the edges of the \clustergraph\ (as there could be
too many edges). Instead, we treat the implicit $k$-decomposition as
an implicit representation of the \clustergraph{}. Given an implicit
representation, our connected components algorithm is the following:
\begin{enumerate}
\item Find a $k$-implicit decomposition of the graph.
\item Run the write-efficient connectivity algorithm from
  Section~\ref{sec:cc-linear} with $\beta=1/k$, treating the
  $k$-decomposition as an implicit representation of the
  \clustergraph{}, i.e., querying edges as needed.
\end{enumerate}

As used in the connectivity algorithm, our implicit representation
need only be able to list the edges adjacent to a center vertex $x$ in
the \clustergraph{}.  To do so, start at $x$, and explore outwards
(e.g., with BFS), keeping all vertices and edges encountered so far in
\local{} memory.  For each frontier vertex $v$, query its center (as
in Lemma~\ref{lemma:findcluster}) --- if $\rho(v) = x$, then $v$'s
unexplored neighbors are added to the next frontier; otherwise (if
$\rho(v) \neq x$) the edge $(x,\rho(v))$ is in the \clustergraph{}, so
add it to the output list.

\begin{lemma}\label{lem:cc-impgraph}
  Assuming a \local{} memory of size $\Omega(k\log n)$, the centers
  neighboring each center in the \clustergraph{} can be listed in no writes and work,
  \depth{}, and operations all $O(k^2)$ in expectation or
  $O(k^2\log n)$ \whp{}.
\end{lemma}
\begin{proof}
  Listing all the vertices in the cluster takes expected work $O(k^2)$
  according to Lemma~\ref{lemma:findcluster}, or $O(k^2\log n)$
  \whp{}.  The number of vertices in the cluster is $O(k)$, so they
  can all fit in \local{} memory. Moreover, since each vertex in the
  cluster has $O(1)$ neighbors, the total number of explored vertices
  in neighboring clusters is $O(k)$, all of which can fit in \local{}
  memory.  Each of these vertices is queried with a cost of $O(k)$
  operations in expectation and $O(k\log n)$ \whp{} given the
  specified \local{} memory (Lemma~\ref{lemma:findcenter}).
\end{proof}

Note that a consequence of the implicit representation is that listing
neighbors is more expensive, and thus the number of operations
performed by a BFS increases, affecting both the work and the
\depth{}.  The implicit representation is only necessary while
operating on the original \clustergraph{}, i.e., while finding the
low-diameter decomposition and spanning trees of each of those vertex
subsets; the contracted graph can be built explicitly as before.  The
best choice of $k$ is $k=\sqrt{\wcost}$, giving us the following
theorem.

\begin{theorem}\label{thm:cc-oracle}
  A connectivity oracle that answers queries in $O(\sqrt{\wcost})$
  expected work and $O(\sqrt{\wcost}\log n)$ work \whp{} can be
  constructed in $O(n/\sqrt{\wcost})$ expected writes,
  $O(\sqrt{\wcost} n)$ expected work, and $O(\wcost^{3/2}\log^3n)$
  \depth{} \whp{} on the \ourmodel\ model, assuming a \local{} memory of
  size $\Omega(\sqrt{\wcost}\log n)$.
\end{theorem}

\begin{proof}
  The $k$-implicit decomposition can be found in $O(n/k)$ writes,
  $O(kn + \wcost n/ k)$ work, and $O(k\log n(k^2\log n + \wcost))$
  \depth{} by Lemmas~\ref{lemma:genclusters}
  and~\ref{lemma:pargenclusters}.  For $k=\sqrt{\wcost}$, these bounds
  reduce to $O(n/\sqrt{\wcost})$ writes, $O(\sqrt{\wcost} n)$ work, and
  $O(\wcost^{3/2}\log^3n)$ \depth{}.

  If we had an explicit representation of the \clustergraph{} with
  $n'=O(n/k)$ vertices and $m' = O(m) = O(n)$ edges, the connectivity
  algorithm would have $O(n'+m'/k) = O(n/k)$ expected writes,
  $O(\wcost n' + \wcost m' / k + m') = O(\wcost n/k + n)$ expected
  work, and $O(\wcost k \log^2n)$ \depth{} \whp{} (by
  Theorem~\ref{thm:cc-linear}).  The fact that the \clustergraph{} is
  implicit means that the BFS needs to perform $O(k^2)$ additional
  work (but not writes) per node in the \clustergraph{}, giving
  expected work $O(\wcost n/k + n + k^2n') = O(\wcost n /k + kn)$.  To
  get a high probability bound, the \depth{} is multiplied by
  $O(k^2\log n)$, giving us $O(\wcost k^3 \log^3 n)$.  For
  $k=\sqrt{\wcost}$, the work and writes match the theorem statement,
  but the \depth{} is larger than claimed by a $\wcost$ factor.

  To remove the extra $\wcost$ factor on the \depth{}, we need to look
  inside the BFS algorithm and its analysis~\cite{BBFGGMS16}.  The
  $O(\wcost {\cal D} \log n)$ \depth{} bound for the BFS, where ${\cal
    D} = O(k\log n)$ is the diameter, is dominated by the \depth{} of
  a packing subroutine on vertices of the frontier.  This packing
  subroutine does not look at edges, and is thus not affected by the
  overhead of finding a vertex's neighbors in the implicit
  representation of the \clustergraph{}.  Ignoring the packing and
  just looking at the exploration itself, the \depth{} of BFS is
  $O({\cal D} \log n)$, which given the implicit representation
  increases by a $O(k^2 \log n)$ factor.  Adding these together, we
  get \depth{} $O(\wcost k \log^2 n + k^3 \log^3 n) =
  O(\wcost^{3/2}\log^3n)$ for the BFS phases.
\end{proof}


We can also output the spanning forest on the contracted graph in the
same bounds, which will be used in the biconnectivity algorithm with
sublinear writes.

\section{Graph Biconnectivity}\label{sec:biconn}

In this section we introduce algorithms related to biconnectivity
and 1-edge connectivity queries.
We first review the classic approach
and its output, which requires $O(m)$ writes.
Then we propose a new BC (biconnected-component) labeling output, which has size $O(n)$ and can be constructed in $O(n)$ writes.
Queries such as determining bridges,
articulation points, and biconnected components can be answered in
$O(1)$ operations (and no writes) with the \imprep{}.  Finally we show how an \implicit{}
(as generated by Algorithm~\ref{algo:genclusters}) can be integrated into
the algorithm to further reduce the writes to $O(n/\sqrt{\wcost})$.

We begin by explaining sequential algorithms that we believe to be new and interesting.
Then in Section~\ref{sec:biconn-depth} we show that these algorithms are parallelizable.
For this section, we assume the size of the \local{}
memory in our model is $O(k\log n)$.

In this section we assume that the graph is connected.
If not, we can run the connectivity algorithm and
then run the algorithm on each component.  The results for a graph
are the union of the results of each of its connected components.

\subsection{The Classic Algorithm}

The classic parallel algorithm~\cite{tarjan1985efficient} to compute
biconnected components and bridges of a connected graph is based on
the Euler-tour technique.  The algorithm starts by building a
spanning tree $T$ rooted at some arbitrary vertex. Each vertex is
labeled by $\first(v)$
and $\last(v)$, which are the ranks of $v$'s first and last appearance
on the Euler tour of $T$.  The low
value $\low(v)$ and the high value $\high(v)$ of a vertex $v\in V$ are
defined as:
\begin{gather*}
\low(v)=\min\{w(u) \mid u \mb{ is in the subtree rooted at } v\}\\
\high(v)=\max\{w(u) \mid u \mb{ is in the subtree rooted at } v\}
\end{gather*}
where
\[w(u)=\min\{\first(u)\cup \{\first(u')\mid (u,u') \mb{ is a nontree
  edge}\}\}\footnote{If there are multiple edges
  $(u,u')$ in the graph, none of them are considered here.}\]

Namely, $\low(v)$ and $\high(v)$ indicate the first and last vertex
on the Euler tour that are connected by a nontree edge to the subtree rooted
at $v$.  The $\low(\cdot)$ and $\high(\cdot)$ values
can be computed by a reduce on each vertex followed by a leaffix\footnote{Leaffix is similar to prefix but defined on a tree and computed from the leaves to the root.} on
the subtrees.  The computation takes $O(\wcost{}\log n)$ \depth,
$O(m+\wcost{}n)$ work, and $O(n)$ writes on the \ourmodel\ model, by using the algorithm
and scheduling theorem in~\cite{BBFGGMS16}.  Then a tree edge is a
bridge if and only if the child's $\low$ and $\high$ is inclusively
within the range of $\first$ and $\last$ of the parent.
This parallel algorithm is asymptotically optimal even sequentially without considering asymmetric read and write costs.

The standard output of biconnected components~\cite{CLRS,JaJa92}
is an array $B[\cdot]$ with size $m$, where the $i$-th element in $B$
indicates which biconnected component the $i$-th edge belongs to.
Explicitly writing-out $B$ is costly in the asymmetric setting,
especially when $m\gg n$.
We provide an alternative
\imprep{} as output that only requires $O(n)$ writes.

\subsection{The BC Labeling}\label{sec:imprep}\label{sec:biconn-linear}

\begin{figure}
\centering
  \includegraphics[width=.5\columnwidth]{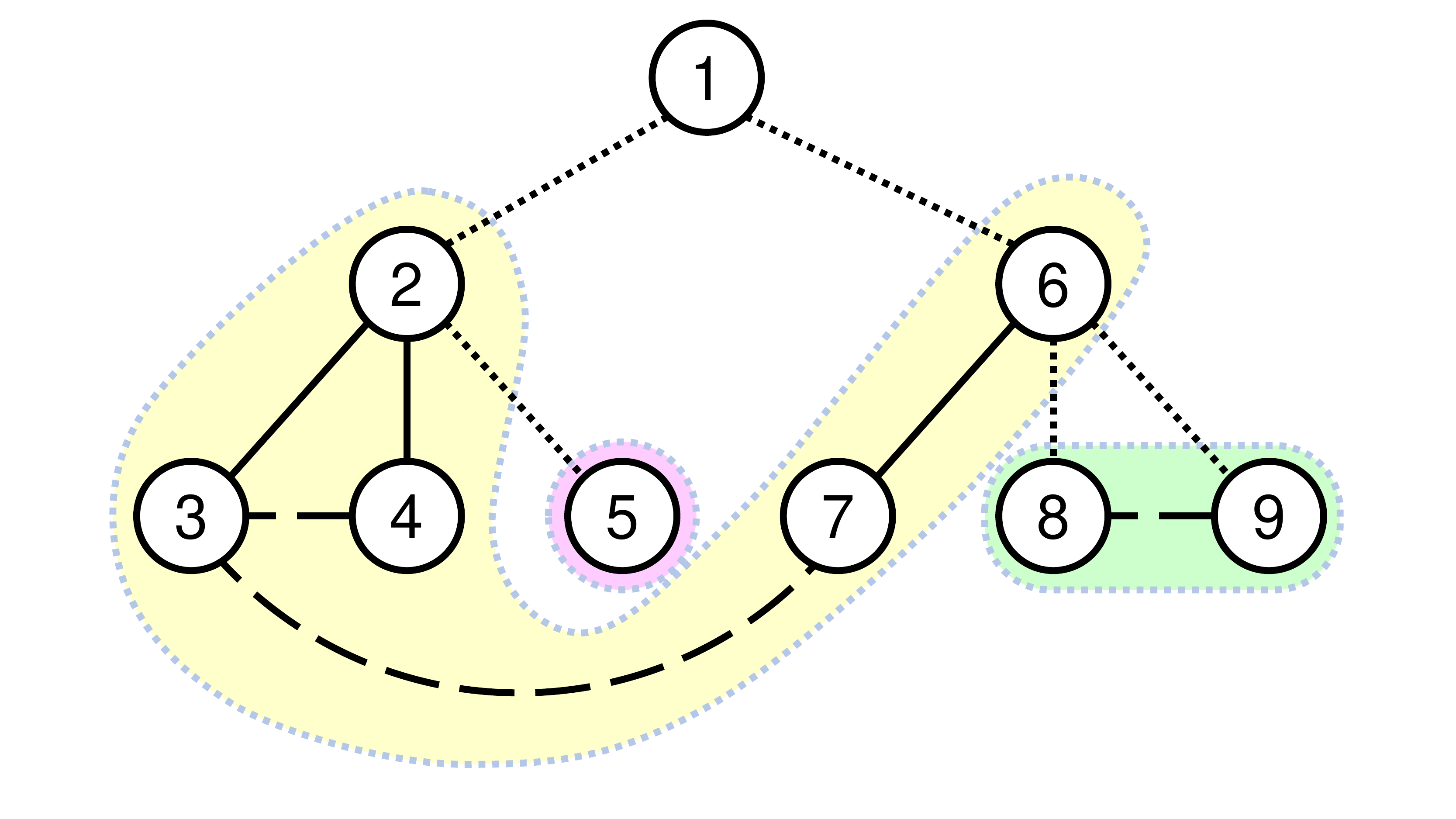}
  \vspace{-1em}
  \caption{An example of the \imprep{} of a graph.  The spanning tree is rooted at vertex 1.  The solid and dot lines indicate tree edges while dot lines are the critical edges.  Dash lines are non-tree edges.  The vertex labels are $l=[1,1,1,2,1,1,3,3]$, and component heads $r=[1,2,6]$.
  Based on the \imprep{} the bridges, articulation points, and
  biconnected components can be easily retrieved as $\{(2,5)\}$,
  $\{2,6\}$, and $\{\{1,2,3,4,6,7\},\{2,5\},\{6,8,9\}\}$.
  }\label{fig:eulertour}
\end{figure}

Here we describe the \defn{BC (biconnected-component) labeling},
which effectively represents biconnectivity output in $O(n)$ space.
Instead of storing all edges within each biconnected
component, the \imprep{} stores a component label for each vertex, and a vertex
for each component.
An example of a BC labeling of a graph is shown in Figure~\ref{fig:eulertour}.
We will later show how to use this representation
along with an implicit decomposition to reduce the writes further.

\begin{definition}[\imprep{}]
  The \imprep{} of a connected graph with respect to a rooted spanning tree
  stores a \defn{vertex label}
  $l:V\backslash \{\textit{root}\}\to [C]$ where $C$ is the number of
  biconnected components in the graph, and a \defn{component head}
  $r:[C]\to V$ of each biconnected component.
\end{definition}

\begin{lemma}
  The \imprep{} of a connected graph can be computed in $O(m)$ operations
  and $O(n + m/\wcost)$ writes on the \seqmodel. Queries about bridges, articulation
  points, or biconnected components can be answered in no writes and $O(1)$ operations
  given a \imprep{} on a rooted spanning tree.
\end{lemma}


\myparagraph{The algorithm to compute \imprep{}} A vertex $v\in V$ (except for the root) is an articulation point iff there exists at least one child $u$ in the spanning tree that has $\first(v)\le\low(u)$ and $\high(u)\le\last(v)$.
We thus name the tree edge between such a pair of vertices to be a \defn{critical edge}.
The algorithm to compute the \imprep{} simply removes all critical edges and runs graph connectivity on all remaining graph edges.
Then the algorithm described in Section~\ref{sec:cc-linear} gives a unique component label that we assign as the vertex label.
For each component, its head is the vertex that is the parent of its cluster in the spanning tree.
Each connected component and its head form a biconnected component.

The correctness of the algorithm can be proven by showing the
equivalence of the result of this algorithm and that of the
Tarjan-Vishkin algorithm~\cite{tarjan1985efficient}.

Since the number of biconnected components is at most $n$, the spanning
tree, vertex labels, and component heads require only linear space. Therefore,
 the space requirement of the \imprep{} is $O(n)$.

\myparagraph{Query on \imprep{}}  We now show that queries are easy with the \imprep{}.
An edge is a bridge iff it is the only edge connecting a single-vertex component and its component head (the biconnected component contains this single edge).
The root of the spanning tree is an articulation point iff it is the head of at least two biconnected components. Any other vertex is an articulation point iff it is the head of at least one biconnected component.
A block-cut tree can also be generated from the \imprep{}: for each vertex create an edge from itself and its vertex label; and for each component create an edge from the label of this component to the component head.  We have a block-cut tree after removing degree-1 nodes corresponding to vertices.

This new representation can be interpreted as an implicit version of the standard
output~\cite{CLRS,JaJa92} of biconnected components,
i.e.\ the label of the biconnected component of each edge can be
reported in $O(1)$ operations.
This is simple: we report the label of the endpoint of the edge that is further from the root along the spanning tree.
The correctness can be shown in two cases: if the edge is a spanning tree edge, then the component label is stored in the further vertex; otherwise, the two vertices must have the same label and reporting either one gives the label of this biconnected component.

\bigskip Using \imprep{} gives the following theorem (see Section~\ref{sec:biconn-depth} for \depth\ analysis).
\begin{theorem}\label{thm:biconn-linear}
Articulation points, bridges, and biconnected components on the \ourmodel\ model take $O(m+n\wcost)$ expected work and $O(\wcost\min\{\wcost,m/n\}\log^2 n)$ \depth\ \whp{}, and each query can be answered in $O(1)$ work.
\end{theorem}

It is interesting to point out that the \imprep{} can efficiently answer queries that are non-trivial when using the standard output.
For example, consider the query: are two vertices in the same biconnected component?  
With the \imprep{} we can answer the query by finding the label of the lower vertex and checking whether the higher one has the same label or is the component head of this component.  To the best of our knowledge, answering such queries on the standard representation can be hard, unless other information is also kept (e.g.\ a block-cut tree).


\subsection{Biconnectivity Oracle in Sublinear Writes}\label{sec:biconn-sublinear}

Next we will show how the \implicit{} generated by Algorithm~\ref{algo:genclusters} can be integrated into the algorithm to further reduce writes in the case of bounded-degree graphs.  Our goal is as follows.

\begin{theorem}\label{thm:biconn}
There exists an algorithm that computes articulation points, bridges, and biconnected components of a bounded-degree graph in $O(n\sqrt{\wcost})$ expected work, $O(n/\sqrt{\wcost{}})$ writes and $O(\wcost^{3/2}\log^3n)$ \depth, and each query takes an expected $O(\wcost{})$ work and $O(\wcost{}\log n)$ work \whp{}, with no writes, on the \ourmodel\ model.
\end{theorem}

The overall idea of the new algorithm is to replace the vertices in the original graph with the clusters generated by Algorithm~\ref{algo:genclusters}.  We generate the \imprep{} on the \clustergraph{} (so the vertex labels are now the \defn{cluster labels}), and then show that a connected-type query can be answered using only the information on the \clustergraph{} and a constant number of associated clusters.
The cost analysis is based on the parameter $k$, and using $k=\sqrt{\wcost{}}$ gives the result in the theorem.


\subsubsection{The \imprep{} on the \clustergraph}

In the first step of the algorithm we generate the \imprep{} on the \clustergraph{}
with $k=\sqrt{\wcost}$.  We root this spanning tree and name it
the clusters spanning tree.
The head vertex of a cluster is chosen as the \defn{cluster root} for that cluster. (The root cluster does not have a cluster root.)
For a cluster, we call the endpoint of a cluster tree edge outside of the cluster an \defn{\outver{}}.
The \defn{\outvers{}} of a cluster is the set of \outvers{}
of all associated cluster tree edges.  Note that all \outvers{} except
for one are the cluster roots for
neighbor clusters.

\subsubsection{The \localgraph{} of a cluster}

We next define the concept of the \localgraph{} of a cluster, so that each query only needs to look up a constant number of associated \localgraph{s}.
An example of a \localgraph{} is shown in Figure~\ref{fig:localgraph} and a more formal definition is as follows.

\begin{figure}
\centering
  \includegraphics[width=.5\columnwidth]{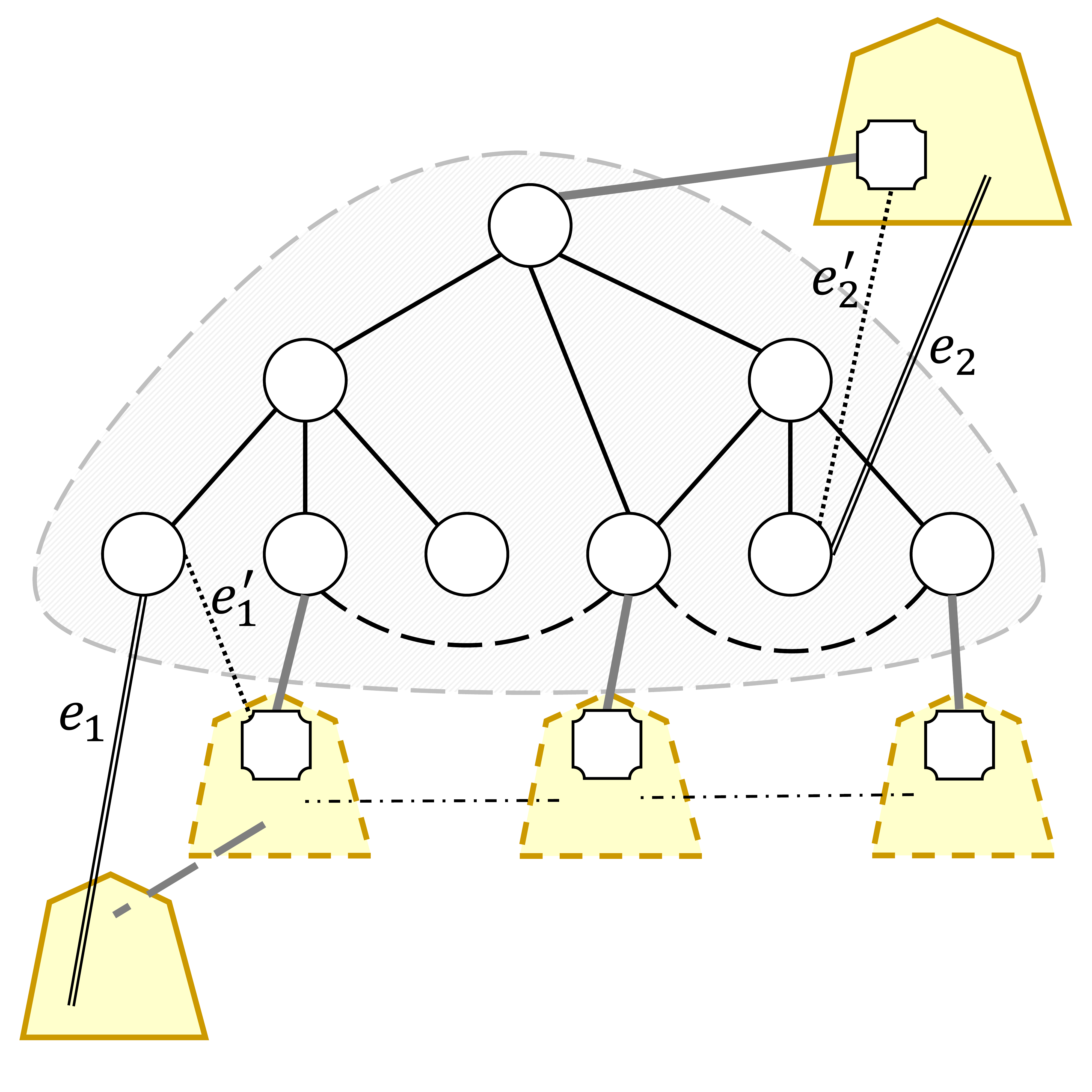}
  \vspace{-1em}
  \caption{An example of a \localgraph{}.  The vertices in the shaded area are in one cluster.  The \localgraph{} contains the vertices in the shaded area and the \outvers{} shown in plaques.  Solid lines indicate the edges that are in the clusters and thick grey lines represent cluster tree edges connecting other clusters (which are shown in yellow pentagons).  The three neighbor clusters sharing the same cluster label are connected using two edges (dash curves).  Edges $e_1$ and $e_2$ are the edges that only has one endpoints in the cluster.  The other endpoint is set to be the \outver{} connecting the cluster of the other original endpoint of this edge in the cluster spanning tree.  Consequently $e_1'$ and $e_2'$ are the replaced edges for $e_1$ and $e_2$.
  }\label{fig:localgraph}
\end{figure}

\begin{definition}[\localgraph]
The \localgraph{} $G'$ of a cluster is defined as $(V_i \cup V_o,E')$. $V_i$ is the set of vertices in the cluster and $V_o$ is the set of \outvers{}. $E'$ consists of:
\begin{enumerate}
  \item The edges with both endpoints in this cluster and the associated clusters' tree edges.
  \item For $c$ neighbor clusters sharing the same cluster label, we find the $c$ corresponding \outvers{} in $V_o$, and connect the vertices with $c-1$ edges.
  \item For an edge $(v_1, v_2)$ with only one endpoint $v_1$ in $V_i$, we find the \outver{} $v_o$ that is connected to $v_2$ on the cluster spanning tree, and create an edge from $v_1$ to $v_o$.
\end{enumerate}
\end{definition}

Figure~\ref{fig:localgraph} shows an example \localgraph{}.
Solid black lines are edges within the cluster and solid grey lines are cluster tree edges.  Neighbor clusters that share a label are shown with dashed outlines and connected via curved dashed lines.  $e_1$ and $e_2$ are examples of edges with only one endpoint in the cluster, and they are replaced by $e_1'$ and $e_2'$ respectively.

Computing \localgraph{s} requires a spanning tree and \imprep{} of the \clustergraph{}.

\begin{lemma}\label{lem:locgraph-cons}
The cost to construct one \localgraph{} is $O(k^2)$ in expectation and $O(k^2\log n)$ \whp{} on the \seqmodel{} model.
\end{lemma}
\begin{proof}
Each cluster in the \implicit{} has at most $k$ vertices, so finding the vertices $V_i$ takes $O(ck)$ cost where $c$ is the cost to compute the mapping $\rho(\cdot)$ of a vertex ($O(k)$ in expectation and $O(k\log n)$ \whp{}).  Since each vertex has a constant degree, there will be at most $O(k)$ neighbor clusters, so $|V_o|=O(k)$.   Enumerating and checking the other endpoint of the edges adjacent to $V_i$ takes $O(ck)$ cost.  Finding the new endpoint of an edge in category 3 requires constant cost after an $O(n/k)$ preprocessing of the Euler tour of the cluster spanning tree.  The number of neighbor clusters is $O(k)$ so checking the cluster labels and adding edges costs no more than $O(k)$.
The overall cost to construct one \localgraph{} is thus $O(k^2)$ in expectation and $O(k^2\log n)$ \whp{}.
Since $c$ is $O(k)$ in expectation and $O(k\log n)$ \whp{}, the overall cost matches the bounds in the lemma.
\end{proof}

\subsubsection{Queries}

With the \localgraph{} and the \imprep{} on the \clustergraph{}, all sorts of biconnectivity queries can be made.
Some of them are easier while other queries require more steps, and the preprocessing steps are shown in an overview of Algorithm~\ref{algo:biconn}.

\begin{algorithm}[t]
\caption{Sublinear-write algorithm for biconnectivity}
\label{algo:biconn}
\newcommand\mycommfont[1]{\footnotesize{#1}}
\SetCommentSty{mycommfont}
\KwIn{Connected bounded-degree graph $G=(V,E)$ and an \implicit{}.}
    \vspace{0.5em}
    Apply connectivity algorithm to generate the \clustergraph.\\
    Compute $\low(\cdot)$ and $\high(\cdot)$ values of all clusters.\\
    Compute the \imprep{} of the \clustergraph{}.\\
    \tcp{Bridges and articulation points can be queried}
    \smallskip
    Compute the root biconnectivity of all \outvers{} in all \localgraph{s}.\\
    Apply leaffix to identify the articulation point of each cluster root.\\
    \tcp{Biconnectivity and 1-edge connectivity on vertices and edges can be queried}
    \smallskip
    Compute the number of biconnected components in each cluster that are completely within this cluster.\\
    Apply prefix sums on the clusters to give an identical label to each biconnected component.\\
    \tcp{The label of biconnected component can be queried}
\end{algorithm}


\myparagraph{Bridges}
There are three cases when deciding whether an edge is a bridge: a tree edge in the clusters spanning tree, a cross edge in the clusters spanning tree, or an edge with both endpoints in the same cluster.
Deciding which case to use takes constant operations.

A tree edge is a bridge if and only if it is a bridge of the \clustergraph{}, which we can mark with $O(n/k)$ writes while computing the \imprep{}. A cross edge cannot be a bridge.

For an edge within a cluster, we use the following lemma:
\begin{lemma}\label{lem:bridge-clust}
An edge with both endpoints in one cluster is a bridge if and only if it is a bridge in the \localgraph{} of the the corresponding cluster.
\end{lemma}
\begin{proof}
If an edge is a bridge in the original graph it means that there are no edges from the subtree of the lower vertex to the outside except for this edge itself.  By applying the modifications of the edges, this property still holds, which means the edge is still a bridge in the \localgraph{} and vice versa.
\end{proof}

Checking if an edge in a cluster is a bridge takes $O(k^2)$ on average and $O(k^2\log n)$ \whp{}.

\myparagraph{Articulation points}  By a similar argument that a vertex is an articulation point of the original graph if and only if it is an articulation point of the associated \localgraph{}.
Given a query vertex $v$, we can check whether it is an articulation point in \localgraph{} associated to $v$, which costs $O(k^2)$ on average and $O(k^2\log n)$ \whp{}.

\bigskip

We now discuss how some more complex queries can be made.  To start with, we show some definitions and results that are used in the algorithms for queries.

\begin{definition}[root biconnectivity]
We say a vertex $v$ in a cluster $C$'s \localgraph{} is root-biconnected if $v$ and the cluster root have the same vertex label in $C$'s \localgraph{}.
\end{definition}

A root-biconnected vertex $v$ indicate that $v$ can connect to the ancestor clusters without using the cluster root (i.e.\ the cluster root is not an articulation point to cut $v$).  Another interpretation is that, there is no articulation point in cluster $C$ that disconnects $v$ from the \outver{} of the cluster root.

\begin{lemma}
Computing and storing the root biconnectivity of all \outvers{} in all
\localgraph{s} takes $O(nk)$ operations in expectation and $O(n/k)$ writes
on the \seqmodel.
\end{lemma}
The proof is straight-forward.  The cost to construct the \localgraph{s} and compute root biconnectivity is $O(nk)$, and since there are $O(n/k)$ clusters tree edges, storing the results requires $O(n/k)$ writes.

\myparagraph{Querying whether two vertices are biconnected}
Checking whether two vertices $v_1$ and $v_2$ can be disconnected by removing any single vertex in the graph is one of the most commonly-used biconnectivity-style queries.
To answer this query, our goal is to find the tree path between this pair of vertices and check whether there is an articulation point on this path that disconnects them.

The simple case is when $v_1$ and $v_2$ are within the same cluster.
We know that the two vertices are connected by a path via the vertices within the cluster. We can check whether any vertex on the path disconnects these two vertices using their vertex labels.

Otherwise, $v_1$ and $v_2$ are in different clusters $C_1$ and $C_2$.  Assume $C_{\smb{LCA}}$ is the cluster that contains the LCA of $v_1$ and $v_2$ (which can be computed by the LCA of $C_1$ and $C_2$ with constant cost) and $v_{\smb{LCA}}\in C_{\smb{LCA}}$ is the LCA vertex.
The tree path between $v_1$ and $v_2$ is from $v_1$ to $C_1$'s cluster root, and then to the cluster root of the \outver{} of $C_1$'s cluster root, and so on, until reaching $v_{\smb{LCA}}$, and the other half of the path can be constructed symmetrically.  It takes $O(k^2)$ expected cost to check whether any articulation point disconnects the paths in $C_1$, $C_2$ and $C_{\smb{LCA}}$.
For the rest of the clusters, since we have already precomputed and stored the root biconnectivity of all outside vertices, then applying a leafix on the clusters spanning tree computes the cluster containing the articulation point of each cluster root.
Therefore checking whether such an articulation point on the path between $C_1$ and $C_{\smb{LCA}}$ or between $C_2$ and $C_{\smb{LCA}}$ that disconnects $v_1$ and $v_2$ takes $O(1)$ cost.  Therefore checking whether two vertices are biconnected requires $O(k^2)$ cost in expectation and no writes.

\hide{
\begin{enumerate}
  \item if one cluster is the ancestor of the other (WLOG assume $C_1$ is the ancestor and $(v_{C_1},v_{C_3})$ is the cluster tree edge on the path between $C_1$ and $C_2$, $v_{C_1}\in C_1$ and $v_{C_3}\in C_3$), then (1) in $C_1$'s \localgraph{} $v_1$ must be biconnected to $v'$, (2) $v_2$ shares the same vertex label with its cluster root, and (3) $C_2$ and $C_3$ have the same cluster label;
  \item otherwise, both vertices share the same vertex labels with the cluster roots in the \localgraph{s} and the two clusters have the same cluster label.
\end{enumerate}
}

\myparagraph{Querying whether two vertices are 1-edge connected}
This is a similar query comparing to the previous one and the only difference is whether an edge, instead of a vertex, is able to disconnect two vertices.
The query can be answered in a similar way by checking whether a bridge disconnects the two vertices on their spanning tree path, which is related to the two clusters containing the two query vertices and the LCA cluster, and the precomputed information for the clusters on the tree path among these three clusters.
The cost for a query is also $O(k^2)$ operations in expectation and it requires no writes.

\myparagraph{Queries on biconnected-component labels for edges}  We now answer the standard queries~\cite{CLRS,JaJa92} of biconnected components: given an edge, report a unique label that represents the biconnected component this edge belongs to.

We have already described the algorithm to check whether any two vertices are biconnected, so the next step is to assign a unique label of each biconnected components, which requires the following lemma:
\begin{lemma}
A vertex in one cluster is either in a biconnected component that only contains vertices in this cluster, or biconnected with at least one \outver{} of this cluster.
\end{lemma}
\begin{proof}
Assume a vertex $v_1$ in this cluster $C$ is biconnected to another vertex $v_2$ outside the cluster, then let $v_o$ be the \outver{} of $C$ on the spanning tree path between $v_1$ and $v_2$, and $v_1$ is biconnected with $v_o$, which proves the lemma.
\end{proof}

With this lemma, we first compute and store the labels of the biconnected components on the cluster roots, which can be finished with $O(nk)$ expected operations and $O(n/k)$ writes with the \imprep{} on the \clustergraph{} and the the root biconnectivity of \outvers{} on all clusters.
Then for each cluster we count the number of biconnected components completely within this cluster.
Finally we apply a prefix sum on the numbers for the clusters to provide a unique label of each biconnected component in every cluster.
Although not explicitly stored, the vertex labels in each cluster can be regenerated with $O(k^2)$ operations in expectation and $O(k^2\log n)$ operations \whp{}, and a vertex label is either the same as that of an \outver{} which is precomputed, or a relative label within the cluster plus the offset of this cluster.

Similar to the algorithm discussed in Section~\ref{sec:imprep}, when a query comes, the edge can either be a cluster tree edge, a cross edge, or within a cluster.
For the first case the label biconnected component is the precomputed label for the (lower) cluster root.
For the second case we just report the vertex label of an arbitrary endpoint, and similarly for the third case the output is the vertex label of the lower vertex in the cluster.
The cost of a query is $O(k^2)$ in expectation and $O(k^2\log n)$ \whp{}.

\bigskip
With the concepts and lemmas in this section, with a precomputation of $O(nk)$ cost and $O(n/k)$ writes, we can also do a normal query with $O(k^2)$ cost in expectation and $O(k^2\log n)$ \whp{} on \textbf{bridge-block tree}, \textbf{cut-block tree}, and \textbf{1-edge-connected components}.

\subsection{Parallelizing Biconnectivity Algorithms}\label{sec:biconn-depth}

The two biconnectivity algorithms discussed in this section are essentially highly parallelizable.
The key algorithmic components include Euler-tour construction, tree contraction, graph connectivity, prefix sum, and preprocessing LCA queries on the spanning tree.
Since the algorithms run each of the components a constant number of times, and the \depth{} of the algorithm is bounded by the \depth{} of graph connectivity, whose bound is provided in Section~\ref{sec:cc} ($O(\wcost^2\log^2 n)$ and $O(\wcost^{3/2}\log^3 n)$ \whp{} respectively when plugging in $\beta$ as $1/\wcost{}$ and $1/\sqrt\wcost{}$).\footnote{The classic parallel algorithms with polylogarithmic depth solve the Euler-tour construction, tree contraction, and prefix sum, since we here only require linear writes (in terms of number of vertices, $O(n)$ and $O(n/k)$ for the two algorithms) for both algorithms.}

For the sublinear-write algorithm, we assume that computations
within a cluster are sequential,
and the work is upper bounded by $O(k^2)=O(\wcost{})$ in expectation and $O(k^2\log n)=O(\wcost{}\log n)$ \whp{} for any computations within a cluster.
This term is additive to the overall depth, since after acquiring the spanning tree (forest) of the clusters, we run all computations within the clusters in parallel and then run tree contraction and prefix sums based on the calculated values.
The $O(\wcost{})$ expected work ($O(\wcost{}\log n)$ \whp{}) is also the cost for a single biconnectivity query, and multiple queries can be done in parallel.
\hide{
The exception is for graph connectivity, where the \depth{} to search the neighbor clusters is multiplicative to the \depth{} of the BFS in the low-diameter decomposition.
Thus the depth of this
algorithm is also bounded by the depth of graph connectivity, which is
$O(\wcost^{3/2}\log^3n)$ \whp{}.
}

\section{Sublinear-Write Algorithms on Unbounded-Degree Graphs}\label{sec:appendix-unbounded}

Here we discuss a solution to generate another graph $G'$ which has bounded degree with $O(m)$ vertices and edges, and the connectivity queries on the original graph $G$ can be answered in $G'$ equivalently.

The overall idea is to build a tree structure with \defn{virtual nodes} for each vertex that has more than a constant degree.
Each virtual node will represent a certain range of the edge list.
Considering a star with all other vertices connecting to a specific vertex $v_1$, we build a binary tree structure with 2 virtual nodes on the first level $v_{1,2\to n/2}$, $v_{1,n/2+1\to n}$, 4 virtual nodes on the second level $v_{1,2\to n/4},\cdots,v_{1,3n/4+1\to n}$ and so on.
We replace the endpoint of an edge from the original graph $G$ with the leaf node in this tree structure that represents the corresponding range with a constant number of edges.
Notice that if both endpoints of an edge have large degrees, then they both have to be redirected.

The simple case is for a sparse graph in which most of the vertices are bounded-degree, and the sum of the degrees for vertices with more than a constant number of edges is $O(n/k)$ (or $O(n/\sqrt{\wcost})$).
In this case we can simply explicitly build a tree structure for the edges of a vertex.

Otherwise, we require the adjacency array of the input graph to have the following property: each edge can query its positions in the edge lists for both endpoints.
Namely, an edge $(u,v)$ knows it is the $i$-th edge in $u$'s edge list and $j$-th edge in $v$'s edge list.
To achieve this, either an extra pointer is stored for each edge, or the edge lists are presorted and the label can be binary searched (this requires $O(\log n)$ work for each edge lookup).
With this information, there exists an implicit graph $G'$ with bounded-degree.
The binary tree structures can be defined such that given an internal tree node, we can find the three neighbors (two neighbors for the root) without explicitly storing the newly added vertices and edges.
Notice that the new graph $G'$ now has $O(m)$ vertices including the virtual ones.
The virtual nodes help to generate \implicit{} and require no writes unless they are selected to be either primary or secondary centers.

Graph connectivity is obviously not affected by this transformation.
It is easy to check that a bridge in the original graph $G$ is also a bridge in the new graph $G'$ and vice versa.
In the biconnectivity algorithm, an edge in $G$ can be split into multiple edges in $G'$, but this will not change the biconnectivity property within a biconnected component, unless the component only contains one bridge edge, which can be checked separately.

This construction, combined with our earlier results, leads to
Theorem~\ref{thm:main-sublinear}.

\section{Conclusion}
This work provides several algorithms targeted at solving graph connectivity
problems considering the read-write asymmetry. Our algorithms make use of
an implicit decomposition technique that is applicable beyond the scope of the
problems studied in this paper. By using this decomposition
and redundantly performing small computation, we are able to reduce the number
of writes in exchange for a small increase in the total number of operations.
This allows us to offset the increased cost of writes in anticipated future systems
and improve overall performance. Even excluding new memory technology, we
believe that research into algorithms with fewer writes provides interesting
results from both a theoretical and memory/cache coherence perspective. Our
work provides a framework which can be used to develop write-efficient solutions
to large graph connectivity problems.








\appendix




\section{Motivation from~\cite{BFGGS15}}\label{sec:hardware}

Further motivation for the asymmetry between reads and write costs in
emerging memory technologies was provided in~\cite{BFGGS15}.  As a
convenience to the reviewer, in this appendix we repeat a suitable
excerpt from that paper.

``While DRAM stores data in capacitors that
typically require refreshing every few milliseconds,
and hence must be continuously powered, emerging NVM
technologies store data as ``states'' of the given material that
require no external power to retain.  Energy is required only to read
the cell or change its value (i.e., its state).  While there is no
significant cost difference between reading and writing DRAM (each
DRAM read of a location not currently buffered requires a write of
the DRAM row being evicted, and hence is also a write),
emerging NVMs such as Phase-Change Memory (PCM), Spin-Torque
Transfer Magnetic RAM (STT-RAM), and Memristor-based Resistive RAM
(ReRAM) each incur significantly higher cost for writing than reading.
This large gap seems fundamental to the technologies themselves: to
change the physical state of a material requires relatively
significant energy for a sufficient duration, whereas reading the
current state can be done quickly and, to ensure the state is left
unchanged, with low energy.  An STT-RAM cell, for example, can be read
in 0.14 $ns$ but uses a 10 $ns$ writing pulse duration, using roughly
$10^{-15}$ joules to read versus $10^{-12}$ joules to
write~\cite{Dong08} (these are the raw numbers at the materials
level).  A Memristor ReRAM cell uses a 100 $ns$ write pulse duration, and
an 8MB Memrister ReRAM chip is projected to have reads with 1.7 $ns$
latency and 0.2 $nJ$ energy versus writes with 200 $ns$ latency and 25 $nJ$
energy~\cite{Xu11}---over two orders of magnitude differences in latency
and energy.  PCM is the most mature of the three technologies, and
early generations are already available as I/O devices.  A recent
paper~\cite{Kim14} reported 6.7 $\mu s$ latency for a 4KB read and
128 $\mu s$ latency for a 4KB write.  Another reported that the
sector I/O latency and bandwidth for random 512B writes was a factor
of 15 worse than for reads~\cite{ibm-pcm14b}.  As a future memory/cache
replacement, a 512Mb PCM memory chip is projected to have 16 $ns$ byte
reads versus 416 $ns$ byte writes, and writes to a 16MB PCM L3 cache
are projected to be up to 40 times slower and use 17 times more energy
than reads~\cite{Dong09}.  While these numbers are speculative and subject
to change as the new technologies emerge over time, there seems to be
sufficient evidence that writes will be considerably more costly than
reads in these NVMs.''

Note that, unlike SSDs and earlier versions of phase-change memory products,
these emerging memory products will sit on the processor's memory bus and be
accessed at byte granularity via loads and stores (like DRAM).  Thus, the
time and energy for reading can be roughly on par with DRAM, and depends
primarily on the properties of the technology itself relative to DRAM.

\section{Formal Definitions of the Terms}\label{sec:appendix-prelim}

A \defn{spanning tree} $T$ of an undirected connected graph $G$ is a subgraph that is a tree which includes all of the vertices of $G$.
A \defn{spanning forest} of $G$ contains the union of the spanning trees of all connected components in $G$.
The \defn{lowest-common-ancestor} (LCA) query for two vertices on a rooted spanning tree requires $O(n)$ work and $O(\log n)$ depth on preprocessing, and $O(1)$ query time~\cite{berkman1993recursive,sadakane2002space}.

A \defn{connected component} of $G$ is a subgraph in which any two vertices are connected to each other by paths via edges in the graph.

A \defn{biconnected component} (also known as a block or 2-connected component) of $G$ is a maximal subgraph such that it is still connected after removing any single vertex in the subgraph.
Any connected graph decomposes into a tree of biconnected components called the block-cut tree of the graph.
The blocks are attached to each other at shared vertices called \defn{articulation points}.

A \defn{bridge} of $G$ is an edge whose deletion increases the number of connected components of the graph.
A connected graph is \defn{$k$-edge-connected} if it remains connected whenever fewer than $k$ edges are removed.
An unconnected graph is 0-edge connected; a connected graph with bridges is 1-edge-connected; and a bridge-less graph is at least 2-edge-connected.



\section{Summary of Low-Diameter Decomposition Algorithm}\label{sec:appendix-cc}

 The algorithm of Miller et al.~\cite{miller2013parallel} generates a
$(\beta, O({\log n}/\beta))$ decomposition with $O(m)$ operations and
$O(\wcost{\log^2 n}/\beta)$ \depth\ \whp{}.  As described by Miller et
al., the number of writes performed is also $O(m)$, but this can be
improved to $O(n)$.  Specifically, the algorithm executes multiple
breadth-first searches (BFS's) in parallel, which can be replaced by
write-efficient BFS's.

In more detail, the algorithm first assigns each vertex $v$ a
value $\delta_v$ drawn from an exponential distribution with parameter
$\beta$ (mean $1/\beta$). Then on iteration $i$ of the algorithm,
BFS's are started from unexplored vertices $v$ where
$\delta_v \in [i,i+1)$ and all BFS's that have already started are
advanced one level. At the end of the algorithm, all vertices that
were visited by a BFS starting from the same source will belong to the
same subset of the decomposition. If a vertex is visited by multiple
BFS's in the same iteration, it can be assigned to an arbitrary
BFS.\footnote{The original analysis of Miller et
  al.~\cite{miller2013parallel} requires the vertex to be assigned to
  the BFS with the smaller fractional part of $\delta_s$, where $s$ is
  the source of the BFS. However, Shun et al.~\cite{Shun2014} show
  that an arbitrary assignment gives the same complexity bounds.}  The
maximum value of $\delta_v$ can be shown to be $O({\log n}/\beta)$
\whp{}, and so the algorithm terminates in $O(\log n/\beta)$
iterations. Each iteration requires $O(\wcost\log n)$ \depth\ for
packing the frontiers of the BFS's, leading to an overall \depth\ of
$O(\wcost\log^2 n/\beta)$ \whp{}. A standard BFS requires operations
and writes that are linear in the number of vertices and edges
explored, giving a total work of $O(\wcost(m+n))$.  By using the
write-efficient BFS from~\cite{BBFGGMS16}, the expected number of
writes for each BFS is proportional to the number of vertices marked
(assigned to it), and so the total expected number of writes is
$O(n)$. Tasks only need $O(1)$ \local{} memory in the algorithm.  This
yields Theorem~\ref{thm:ldd}.

\hide{
\begin{proof}[{\bf Proof of Theorem~\ref{thm:cc-oracle}}]
  The $k$-implicit decomposition can be found in $O(n/k)$ writes,
  $O(kn + \wcost n/ k)$ work, and $O(k\log n(k^2\log n + \wcost))$
  \depth{} by Lemmas~\ref{lemma:genclusters}
  and~\ref{lemma:pargenclusters}.  For $k=\sqrt{\wcost}$, these bounds
  reduce to $O(n/\sqrt{\wcost})$ writes, $O(\sqrt{\wcost} n)$ work, and
  $O(\wcost^{3/2}\log^3n)$ \depth{}.

  If we had an explicit representation of the \clustergraph{} with
  $n'=O(n/k)$ vertices and $m' = O(m) = O(n)$ edges, the connectivity
  algorithm would have $O(n'+m'/k) = O(n/k)$ expected writes,
  $O(\wcost n' + \wcost m' / k + m') = O(\wcost n/k + n)$ expected
  work, and $O(\wcost k \log^2n)$ \depth{} \whp{} (by
  Theorem~\ref{thm:cc-linear}).  The fact that the \clustergraph{} is
  implicit means that the BFS needs to perform $O(k^2)$ additional
  work (but not writes) per node in the \clustergraph{}, giving
  expected work $O(\wcost n/k + n + k^2n') = O(\wcost n /k + kn)$.  To
  get a high probability bound, the \depth{} is multiplied by
  $O(k^2\log n)$, giving us $O(\wcost k^3 \log^3 n)$.  For
  $k=\sqrt{\wcost}$, the work and writes match the theorem statement,
  but the \depth{} is larger than claimed by a $\wcost$ factor.

  To remove the extra $\wcost$ factor on the \depth{}, we need to look
  inside the BFS algorithm and its analysis~\cite{BBFGGMS16}.  The
  $O(\wcost {\cal D} \log n)$ \depth{} bound for the BFS, where ${\cal
    D} = O(k\log n)$ is the diameter, is dominated by the \depth{} of
  a packing subroutine on vertices of the frontier.  This packing
  subroutine does not look at edges, and is thus not affected by the
  overhead of finding a vertex's neighbors in the implicit
  representation of the \clustergraph{}.  Ignoring the packing and
  just looking at the exploration itself, the \depth{} of BFS is
  $O({\cal D} \log n)$, which given the implicit representation
  increases by a $O(k^2 \log n)$ factor.  Adding these together, we
  get \depth{} $O(\wcost k \log^2 n + k^3 \log^3 n) =
  O(\wcost^{3/2}\log^3n)$ for the BFS phases.
\end{proof}
}

\bibliographystyle{plain}
\bibliography{../../ref}

\begin{thebibliography}{10}

\bibitem{abraham2012using}
Ittai Abraham and Ofer Neiman.
\newblock Using petal-decompositions to build a low stretch spanning tree.
\newblock In {\em Proceedings of the forty-fourth annual ACM symposium on
  Theory of computing}, pages 395--406. ACM, 2012.

\bibitem{andreev2004balanced}
Konstantin Andreev and Harald R{\"a}cke.
\newblock Balanced graph partitioning.
\newblock In {\em Proceedings of the ACM Symposium on Parallelism in Algorithms
  and Architectures}, pages 120--124, 2004.

\bibitem{Arulraj17sigmod}
Joy Arulraj and Andrew Pavlo.
\newblock How to build a non-volatile memory database management system.
\newblock In {\em {ACM} International Conference on Management of Data,
  {SIGMOD} Conference 2017}, pages 1753--1758, 2017.

\bibitem{awerbuch1985complexity}
Baruch Awerbuch.
\newblock Complexity of network synchronization.
\newblock {\em Journal of the ACM (JACM)}, 32(4):804--823, 1985.

\bibitem{awerbuch1992low}
Baruch Awerbuch, Bonnie Berger, Lenore Cowen, and David Peleg.
\newblock Low-diameter graph decomposition is in {NC}.
\newblock In {\em Scandinavian Workshop on Algorithm Theory}, pages 83--93.
  Springer, 1992.

\bibitem{awerbuch1989network}
Baruch Awerbuch, M~Luby, AV~Goldberg, and Serge~A Plotkin.
\newblock Network decomposition and locality in distributed computation.
\newblock In {\em Foundations of Computer Science, 1989., 30th Annual Symposium
  on}, pages 364--369. IEEE, 1989.

\bibitem{Bausch2012damon}
Daniel Bausch, Ilia Petrov, and Alejandro Buchmann.
\newblock Making cost-based query optimization asymmetry-aware.
\newblock In {\em ACM DaMoN}, 2012.

\bibitem{BT06}
Avraham Ben-Aroya and Sivan Toledo.
\newblock Competitive analysis of flash-memory algorithms.
\newblock In {\em Proc.~European Symposium on Algorithms (ESA)}, 2006.

\bibitem{BBFGGMS16}
Naama Ben-David, Guy~E. Blelloch, Jeremy~T. Fineman, Phillip~B. Gibbons, Yan
  Gu, Charles McGuffey, and Julian Shun.
\newblock Parallel algorithms for asymmetric read-write costs.
\newblock In {\em Proc. ACM Symposium on Parallelism in Algorithms and
  Architectures (SPAA)}, 2016.

\bibitem{Berenbrink2014ipl}
Petra Berenbrink, Bruce Krayenhoff, and Frederk Mallmann-Trenn.
\newblock Estimating the number of connected components in sublinear time.
\newblock {\em Infomrtaion Processing Letters}, 114(11), 2014.

\bibitem{berkman1993recursive}
Omer Berkman and Uzi Vishkin.
\newblock Recursive star-tree parallel data structure.
\newblock {\em SIAM Journal on Computing}, 22(2):221--242, 1993.

\bibitem{BFGGS15}
Guy~E. Blelloch, Jeremy~T. Fineman, Phillip~B. Gibbons, Yan Gu, and Julian
  Shun.
\newblock Sorting with asymmetric read and write costs.
\newblock In {\em Proc. ACM Symposium on Parallelism in Algorithms and
  Architectures (SPAA)}, 2015.

\bibitem{blelloch2016efficient}
Guy~E. Blelloch, Jeremy~T. Fineman, Phillip~B. Gibbons, Yan Gu, and Julian
  Shun.
\newblock Efficient algorithms with asymmetric read and write costs.
\newblock In {\em 24th Annual European Symposium on Algorithms}, pages
  14:1--14:18, 2016.

\bibitem{blelloch2014nearly}
Guy~E Blelloch, Anupam Gupta, Ioannis Koutis, Gary~L Miller, Richard Peng, and
  Kanat Tangwongsan.
\newblock Nearly-linear work parallel sdd solvers, low-diameter decomposition,
  and low-stretch subgraphs.
\newblock {\em Theory of Computing Systems}, 55(3):521--554, 2014.

\bibitem{carson2016write}
Erin Carson, James Demmel, Laura Grigori, Nicholas Knight, Penporn
  Koanantakool, Oded Schwartz, and Harsha~Vardhan Simhadri.
\newblock Write-avoiding algorithms.
\newblock In {\em IEEE International Parallel and Distributed Processing
  Symposium}, pages 648--658, 2016.

\bibitem{chazelle2005siam}
Bernard Chazelle, Ronitt Rubinfeld, and Luca Trevisan.
\newblock Approximating the minimum spanning tree weight in sublinear time.
\newblock {\em SIAM J. Comput.}, 34(6), 2005.

\bibitem{Chen11}
Shimin Chen, Phillip~B. Gibbons, and Suman Nath.
\newblock Rethinking database algorithms for phase change memory.
\newblock In {\em Proc.~Conference on Innovative Data Systems Research (CIDR)},
  2011.

\bibitem{Chen2015pvldb}
Shimin Chen and Qin Jin.
\newblock Persistent {B}+-trees in non-volatile main memory.
\newblock {\em Proceedings of the VLDB Endowment}, 8(7):786--797, 2015.

\bibitem{ChoL09}
Sangyeun Cho and Hyunjin Lee.
\newblock {Flip-N-Write}: A simple deterministic technique to improve {PRAM}
  write performance, energy and endurance.
\newblock In {\em Proc.~IEEE/ACM International Symposium on Microarchitecture
  (MICRO)}, 2009.

\bibitem{ColeKT96}
Richard Cole, Philip~N. Klein, and Robert~Endre Tarjan.
\newblock Finding minimum spanning forests in logarithmic time and linear work
  using random sampling.
\newblock In {\em SPAA}, pages 243--250, 1996.

\bibitem{CLRS}
Thomas~H. Cormen, Charles~E. Leiserson, Ronald~L. Rivest, and Clifford Stein.
\newblock {\em Introduction to Algorithms (3rd edition)}.
\newblock MIT Press, 2009.

\bibitem{Dong09}
Xiangyu Dong, Norman~P. Jouupi, and Yuan Xie.
\newblock {PCRAMsim}: System-level performance, energy, and area modeling for
  phase-change {RAM}.
\newblock In {\em Proc.~ACM International Conference on Computer-Aided Design
  (ICCAD)}, 2009.

\bibitem{Dong08}
Xiangyu Dong, Xiaoxia Wu, Guangyu Sun, Yuan Xie, Hai~H. Li, and Yiran Chen.
\newblock Circuit and microarchitecture evaluation of {3D} stacking magnetic
  {RAM (MRAM)} as a universal memory replacement.
\newblock In {\em Proc.~ACM Design Automation Conference (DAC)}, 2008.

\bibitem{Eppstein14}
David Eppstein, Michael~T. Goodrich, Michael Mitzenmacher, and Pawel Pszona.
\newblock Wear minimization for cuckoo hashing: How not to throw a lot of eggs
  into one basket.
\newblock In {\em Proc.~ACM International Symposium on Experimental Algorithms
  (SEA)}, 2014.

\bibitem{Gal05}
Eran Gal and Sivan Toledo.
\newblock Algorithms and data structures for flash memories.
\newblock {\em ACM Computing Surveys}, 37(2), 2005.

\bibitem{Gazit1991}
Hillel Gazit.
\newblock An optimal randomized parallel algorithm for finding connected
  components in a graph.
\newblock {\em SIAM J. Comput.}, 20(6):1046--1067, December 1991.

\bibitem{HalperinZ96}
Shay Halperin and Uri Zwick.
\newblock An optimal randomized logarithmic time connectivity algorithm for the
  {EREW PRAM}.
\newblock {\em J. Comput. Syst. Sci.}, 53(3):395--416, 1996.

\bibitem{Halperin00}
Shay Halperin and Uri Zwick.
\newblock Optimal randomized {EREW PRAM} algorithms for finding spanning
  forests.
\newblock In {\em J. Algorithms}, pages 1740--1759, 2000.

\bibitem{HK81}
Jia{-}Wei Hong and H.~T. Kung.
\newblock {I/O} complexity: The red-blue pebble game.
\newblock In {\em Proc. {ACM} Symposium on Theory of Computing (STOC)}, 1981.

\bibitem{ibm-pcm14b}
www.slideshare.net/IBMZRL/theseus-pss-nvmw2014, 2014.

\bibitem{jacob2017}
Rico Jacob and Nodari Sitchinava.
\newblock Lower bounds in the asymmetric external memory model.
\newblock In {\em Proceedings of the 29th ACM Symposium on Parallelism in
  Algorithms and Architectures}, SPAA '17, 2017.

\bibitem{JaJa92}
J.~Jaja.
\newblock {\em Introduction to Parallel Algorithms}.
\newblock Addison-Wesley Professional, 1992.

\bibitem{Kim14}
Hyojun Kim, Sangeetha Seshadri, Clement~L. Dickey, and Lawrence Chu.
\newblock Evaluating phase change memory for enterprise storage systems: A
  study of caching and tiering approaches.
\newblock In {\em Proc.~USENIX Conference on File and Storage Technologies
  (FAST)}, 2014.

\bibitem{LeeIMB09}
Benjamin~C. Lee, Engin Ipek, Onur Mutlu, and Doug Burger.
\newblock Architecting phase change memory as a scalable {DRAM} alternative.
\newblock In {\em Proc.~ACM International Symposium on Computer Architecture
  (ISCA)}, 2009.

\bibitem{linial1991decomposing}
Nathan Linial and Michael~E Saks.
\newblock Decomposing graphs into regions of small diameter.
\newblock In {\em SODA}, volume~91, pages 320--330, 1991.

\bibitem{miller2013parallel}
Gary~L Miller, Richard Peng, and Shen~Chen Xu.
\newblock Parallel graph decompositions using random shifts.
\newblock In {\em Proceedings of the ACM symposium on Parallelism in Algorithms
  and Architectures}, pages 196--203, 2013.

\bibitem{Oukid2016sigmod}
Ismail Oukid, Johan Lasperas, Anisoara Nica, Thomas Willhalm, and Wolfgang
  Lehner.
\newblock {FPTree: A Hybrid SCM-DRAM Persistent and Concurrent B-Tree for
  Storage Class Memory}.
\newblock In {\em ACM SIGMOD}, pages 371--386, 2016.

\bibitem{ParkS09}
Hyoungmin Park and Kyuseok Shim.
\newblock {FAST}: Flash-aware external sorting for mobile database systems.
\newblock {\em Journal of Systems and Software}, 82(8), 2009.

\bibitem{PettieR02}
Seth Pettie and Vijaya Ramachandran.
\newblock A randomized time-work optimal parallel algorithm for finding a
  minimum spanning forest.
\newblock {\em SIAM J. Comput.}, 31(6):1879--1895, 2002.

\bibitem{PoonR97}
Chung~Keung Poon and Vijaya Ramachandran.
\newblock A randomized linear work {EREW PRAM} algorithm to find a minimum
  spanning forest.
\newblock In {\em ISAAC}, pages 212--222, 1997.

\bibitem{rosenberg2001graph}
A.L. Rosenberg and L.S. Heath.
\newblock {\em Graph Separators, with Applications}.
\newblock Frontiers in Computer Science. Springer US, 2001.

\bibitem{sadakane2002space}
Kunihiko Sadakane.
\newblock Space-efficient data structures for flexible text retrieval systems.
\newblock In {\em International Symposium on Algorithms and Computation}, pages
  14--24. Springer, 2002.

\bibitem{Shun2014}
Julian Shun, Laxman Dhulipala, and Guy Blelloch.
\newblock A simple and practical linear-work parallel algorithm for
  connectivity.
\newblock In {\em Proceedings of the ACM Symposium on Parallelism in Algorithms
  and Architectures}, pages 143--153, 2014.

\bibitem{swendsen1987nonuniversal}
Robert~H Swendsen and Jian-Sheng Wang.
\newblock Nonuniversal critical dynamics in monte carlo simulations.
\newblock {\em Physical review letters}, 58(2):86, 1987.

\bibitem{tarjan1985efficient}
Robert~E Tarjan and Uzi Vishkin.
\newblock An efficient parallel biconnectivity algorithm.
\newblock {\em SIAM Journal on Computing}, 14(4):862--874, 1985.

\bibitem{Viglas12}
Stratis~D. Viglas.
\newblock Adapting the {B}$^+$-tree for asymmetric {I/O}.
\newblock In {\em Proc.~East European Conference on Advances in Databases and
  Information Systems (ADBIS)}, 2012.

\bibitem{Viglas14}
Stratis~D. Viglas.
\newblock Write-limited sorts and joins for persistent memory.
\newblock {\em {PVLDB}}, 7(5), 2014.

\bibitem{wang2013wade}
Zhe Wang, Shuchang Shan, Ting Cao, Junli Gu, Yi~Xu, Shuai Mu, Yuan Xie, and
  Daniel~A Jim{\'e}nez.
\newblock Wade: Writeback-aware dynamic cache management for nvm-based main
  memory system.
\newblock {\em ACM Transactions on Architecture and Code Optimization (TACO)},
  10(4):51, 2013.

\bibitem{Xu11}
Cong Xu, Xiangyu Dong, Norman~P. Jouppi, and Yuan Xie.
\newblock Design implications of memristor-based {RRAM} cross-point structures.
\newblock In {\em Proc.~IEEE Design, Automation and Test in Europe (DATE)},
  2011.

\bibitem{yang:iscas07}
Byung-Do Yang, Jae-Eun Lee, Jang-Su Kim, Junghyun Cho, Seung-Yun Lee, and
  Byoung-Gon Yu.
\newblock A low power phase-change random access memory using a data-comparison
  write scheme.
\newblock In {\em Proc.~IEEE International Symposium on Circuits and Systems
  (ISCAS)}, 2007.

\bibitem{zhou2012writeback}
Miao Zhou, Yu~Du, Bruce Childers, Rami Melhem, and Daniel Moss{\'e}.
\newblock Writeback-aware partitioning and replacement for last-level caches in
  phase change main memory systems.
\newblock {\em ACM Transactions on Architecture and Code Optimization (TACO)},
  8(4):53, 2012.

\bibitem{ZhouZYZ09}
Ping Zhou, Bo~Zhao, Jun Yang, and Youtao Zhang.
\newblock A durable and energy efficient main memory using phase change memory
  technology.
\newblock In {\em Proc.~ACM International Symposium on Computer Architecture
  (ISCA)}, 2009.

\bibitem{ZWT13}
Omer Zilberberg, Shlomo Weiss, and Sivan Toledo.
\newblock Phase-change memory: An architectural perspective.
\newblock {\em ACM Computing Surveys}, 45(3), 2013.

\end{thebibliography}

\end{document}